\pdfoutput=1

\documentclass[journal]{IEEEtran}
\usepackage{cite}
\usepackage{color}
\usepackage{multirow} 
\definecolor{Yellow}{rgb}{1,1,0}

\newcommand\defeq{\stackrel{\mathrm{def}}{=}}

\usepackage{graphicx}
\graphicspath{{figures/}}
\DeclareGraphicsExtensions{.pdf,.jpeg,.png,.tif}
\usepackage[cmex10]{amsmath}
\usepackage{amssymb}
\usepackage{algorithmic,algorithm}
\newtheorem{theorem}{Theorem}
\newtheorem{corollary}[theorem]{Corollary}
\newtheorem{lemma}[theorem]{Lemma}
\newtheorem{definition}[theorem]{Definition}
\usepackage{graphicx}
\newcommand{\qed}{\nobreak \ifvmode \relax \else
      \ifdim\lastskip<1.5em \hskip-\lastskip
      \hskip1.5em plus0em minus0.5em \fi \nobreak
      \vrule height0.75em width0.5em depth0.25em\fi}

\usepackage{algorithmic}
\usepackage{array}
\usepackage{mdwmath}
\usepackage{mdwtab}
\usepackage{eqparbox}

\DeclareMathOperator*{\argmax}{argmax}
\usepackage{url}
\begin{document}
\title{Compressed Genotyping }
%\author{Y. Erlich, A. Gordon, M. Brand, G.J. Hannon, and P.P. Mitra}

%%%%%%%%%%%%%%IEEEE_CHAGE
 \author{Yaniv~Erlich, Assaf~Gordon, Michael Brand, Gregory J. Hannon and~Partha~P.~Mitra}

%%%%%%%%%%%%IEEE_CHANGE
\thanks{
 	Y. E, A.G, G.J.H, and P.P.M are with the Watson School of Biological Science, Cold Spring Harbor Laboratory, NY, 11724 USA 
	M.B. is with Lester Associates, Bentleigh East, 3165 Australia
 	}%
\thanks{Email addresses: Y.E is in erlich@cshl.edu, A.G is in gordon@cshl.edu, M.B is in ieee@brand.scso.com, G.J.H is in hannon@cshl.edu, and P.P.M is in mitra@cshl.edu}%
\thanks{Manuscript received May 17, 2009; revised Month dd, yyyy.}%

\markboth{Submitted to Transaction on Information Theory - Special Issue on Molecular Biology and Neuroscience}%
%\markboth{Transaction of Information Theory,~Vol.~X, No.~X, Month~yyyy}%
{Shell \MakeLowercase{\textit{et al.}}:}

%%%%%%%%%%%%%%%%%STOP_CHANGE
\maketitle
\begin{abstract}
Significant volumes of knowledge have been accumulated in recent 
years linking subtle genetic variations to a wide variety of medical disorders from Cystic Fibrosis to mental retardation. 
Nevertheless, there are still great challenges in
applying this knowledge routinely in the clinic, largely due to the relatively 
tedious and expensive process of DNA sequencing.
Since the genetic polymorphisms that underlie these disorders are 
relatively rare in the human population, the presence or absence of a disease-linked polymorphism can be thought of as a sparse signal. Using methods and ideas from compressed sensing and group testing, 
we have developed a cost-effective genotyping protocol. 
In particular, we have adapted our scheme to a recently developed class of 
high throughput DNA sequencing technologies, and assembled a mathematical 
framework that has some important distinctions from 'traditional'
compressed sensing ideas in order to address different biological and technical constraints.

\end{abstract}

\section{INTRODUCTION}

%%%%IEEE_CHANGE
Genotyping, the process of determining the genetic variation of a certain trait in an individual, 
has become a pivotal component of medical genetics, as a broad spectrum
of disorders are now known to be induced by non-functional genes.
In the past thirty years, extensive efforts were made to identify and locate risk
alleles of severe genetic diseases, which are characterized by incapacities or lethality
of the affected individuals at an early age, and very low prevalence in the population.
These efforts have not only led to deeper insights regarding the molecular
mechanisms that underlie those genetic disorders, but have also contributed
to the emergence of large scale genetic screens, where individuals are
genotyped for a panel of risk alleles in order to detect genetic disorders and
provide early intervention where possible.

Genetic diseases are broadly classified into two groups according to the
effect of the underlying mutation - either {\it dominant} or {\it recessive}.
This classification is elucidated by the diploidy of the human genome, 
meaning that each gene appears in two copies (except the X and Y chromosomes in male). Dominant 
mutation induces a disorder even when present only on one chromosome, 
whereas recessive mutation induces the disorder only if both copies are non-functional.
Thus, for a disease caused by a recessive mutation, individuals are classified into three groups : 
(a) {\it normal} if their two alleles are intact, (b) {\it carriers} if only one allele is functional 
(c) {\it affected} if their two alleles are non-functional. Table \ref{table_genotype_phentype} illustrates
this classification. A {\it carrier screen} is a genetic test for
detecting individuals that are heterozygous with respect to a risk allele of a severe genetic disease.
If two carriers bread, they have $25\%$ chance of having an affected offspring for each carriage, and
in some countries that face high prevalence of severe genetic diseases the
practice is to offer a screen to the entire population, regardless their familial history 
for early monitoring and prevention \cite{article_screen_in_other_places, article_screens_in_israel}.
Therefore, due to the importance and ubiquity in medical gentics and the intriguing connection
to the theory of sparse signal recovery, our work is focused on carrier screens.
However, large parts of the framework can be used also for other types of genetic screens.

The most common genotyping method is based on sequencing the regions that encompass 
the risk genes and analyzing the type of the DNA sequence - whether it matches the wild type (WT) 
alleles or a known risk allele. This approach gained popularity 
due to its high accuracy (sensitivity and specificy), 
applicability to a wide variety of genetic disorders, and technical simplicity.
However, the current DNA sequencing technologies used for genotyping 
provide only serial processing of one specimen/region combination at a time. This increases
the cost of labor and other expenses in large-scale screens and essentially deters 
individual participation and limits the panel of genes that are analyzed.

Recently, a new class of DNA sequencer methods, dubbed {\it next-generation sequencing technologies}
(NGST) has emerged, revolutionizing molecular biology and genomics 
(reviewed in \cite{article_metzker, article_the_year_of_sequencing, article_nature_biotech_sequencing_review}). 
These sequencers process the DNA fragments in parallel and 
provide millions of sequence reads in a single batch, each of which corresponds 
to a DNA molecule within the sample. While there are several types of NGST platforms and
different sets of sequencing reactions, all platforms achieve parallelization using a common 
concept of immobilizing the DNA fragments to a surface, so that each fragment occupies a distinct 
spatial position.  When the sequencing regants are applied 
to the surface, they generate optical signals dependent on
the DNA sequence, which are then captured by a microscope and processed. 
Since the fragments are immobilized, successive signals from the same spatial location 
convey the DNA sequence of the corresponding fragment (Fig. \ref{fig_dna_and_labels}a).
However, the spatial locations of the DNA fragments are 
completely random and are based on stochastic hybridization of a small
aliquot (millions) of DNA fragments out of significantly larger number present in a sample.
Therefore, if a DNA library is 
composed of multiple specimens, it is not possible to associate the sequence reads with 
their corresponding specimens. This limitation is the main obstacle to the utilization of
next generation sequencers in large scale screens, since dedicating a run to each 
specimen is not cost effective.

A simple solution to overcome the specimen-multiplexing problem is to append 
unique identifiers, dubbed {\it DNA barcodes}, to each specimen prior to sequencing
\cite{article_barcoding_1}\cite{article_barcoding_2}.
These barcodes are short DNA molecules that are artificially synthesized, and when attached
to the DNA fragment, they label it with a unique sequence. The sequencer reads the entire fragment,
and reports the sequence of the barcode with the sequence of the interrogated region.
By reading the portion of the sequence corresponding to the barcode, the experimenter can associate a genotyped fragment to a given specimen (Fig. \ref{fig_dna_and_labels}b).
While this method was found quite successful 
for genotyping several dozens of specimens, the synthesis and ligation 
of large number of DNA fragments is both cumbersome and expensive.
This restricts the scalability of the method for genetic screens that consist of
thousands of individuals. In fact, with the current costs of synthesizing so 
many DNA barcodes, it is more cost beneficial to use the legacy serial-based DNA sequencers.

Drawing inspiration from compressed sensing
\cite{article_compressed_sensing_background_1, article_compressed_sensing_background_2}, we ask:
{\it since only a small fraction of the population are
carriers of a severe genetic disease, can one employ a compressed genotyping protocol to identify
those individuals?} We suggest a protocol in which one genotypes pools of specimens on a next-generation
sequencing platform that would approach the sequencing capacity, while reducing the 
number of barcodes and maintaining a faithful detection of the carriers.

\begin{table}[!t]
\caption{Genotype-Phenotype Connections in Genetic Disorders}
\centering
\begin{tabular}{ | c || p{2.2cm} || p{1.4cm} || p{2.2cm} | }%
\hline
\bf{Alleles} & \bf{Genotype} & \bf{Dominant Disorder}  & \bf{Recessive Disorder} \\[1ex]
\hline
&&&\\
\bf{AA} & Homozygous WT & Normal    & Normal   \\[1ex]
\bf{Aa} & Heterozygous        & Affected  & Carrier   \\[1ex]
\bf{aa} & Homozygous mut. & Affected  & Affected  \\[1ex]
\hline
\multicolumn{3}{l} {A  - normal allele,  a  - mutant allele}
\end{tabular}
\label{table_genotype_phentype}
\end{table}

\begin{figure*}[!t]
\centering
\includegraphics[width=5.5in]{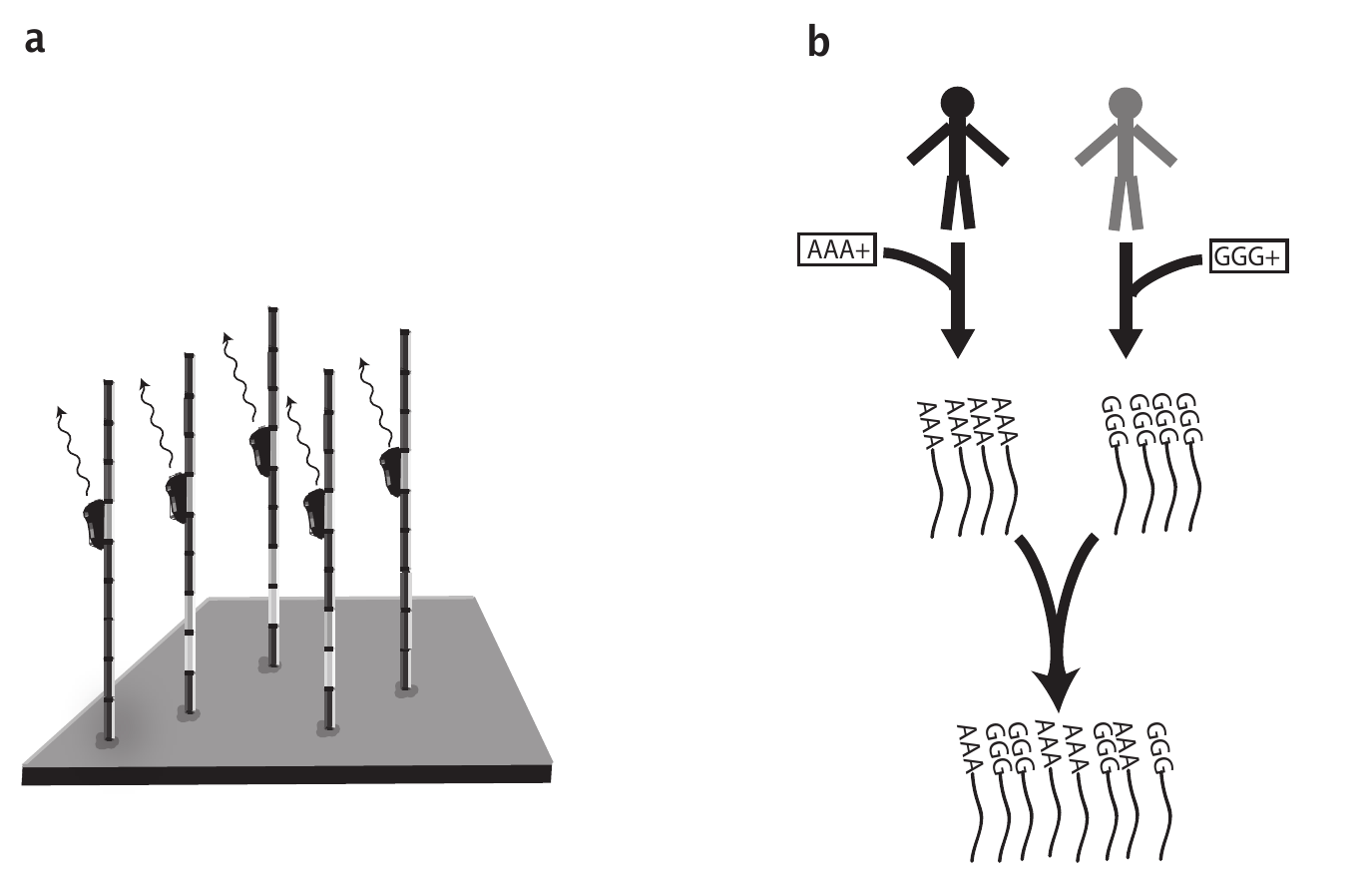}
\caption{Common Techniques in Next Generation Sequencing
(a) High throughput sequencing is employed by immobilizing DNA fragments (rods) on a slide,
and using sequencing reagents (black ovals) that generate optical signals. (b) DNA barcoding
is based on synthesizing short DNA sequences, in the example 'AAA' and 'GGG', and ligating
them to the samples in distinct reactions. When the samples are mixed the barcodes maintain
the identities of the specimens.}
\label{fig_dna_and_labels}
\end{figure*}

\subsection {Related work}
 Our work closely relates to group testing and compressed sensing, which
deals with efficient methods for extracting sparse information from a small 
number of aggregated measurements. Much of the group testing
literature (thoroughly reviewed in: \cite{book_group_testing_1, book_group_testing_2}) is dedicated 
to the {\it prototypical problem}, which describes a set of 
interrogated items that can be in an active state or an inactive state
and a test procedure, which is performed on pools of items, and returns 
'inactive' if all items in the pool are inactive, or 'active' if at least one of the items in the pool is active.
Mathematically, this type of test can be thought of as an OR operation over the items in the pool, and is
called {\it superimposition}\cite{article_kautz_singelton}. 
In general, there are two types of test schedules: adaptive schedules, in 
which items are analyzed in successive rounds and re-pooled from round to round according
to the accumulated results, and non-adaptive schedules where the items are pooled and tested
in a single round. While in theory adaptive schedules require less tests, in practice they are 
more labor intensive and time consuming due to the re-pooling steps and the need to wait for 
the test results from the previous round. For that reason, non-adaptive schedules are favored, 
and have been employed for several biological applications including 
finding sequencing tagged sites in yeast artificial chromosomes (YAC)\cite{article_bruno}
, and mapping protein interactions\cite{atricle_y2h}.

Compressed sensing\cite{article_compressed_sensing_background_1, article_compressed_sensing_background_2}
is a recently emerged signal processing technique that describes conditions 
and efficient methods for capturing signals that are sparse in some 
orthonormal representation by measuring a small number of linear projections. 
This theory extends the framework of group testing to the recovery of hidden variables that are real 
(or complex) numbers. Additonal deviation from group testing is that the aggregated measurements
reports the linear combination of the data points and not superimposition. 
However, some combinatorial concepts in group testing were found useful also for compressed sensing,
and it has been recently shown that deterministic designs based on group testing can confer
sublinear reconstruction runtimes for real signals
\cite{article_cormode,article_sudocodes,article_group_testing_compressed_sensing}.
The framework of compressed
sensing was also found useful for applications beyond 'traditional' signal processing
and recently a novel biological application was suggested - designing highly efficient microarrays that 
reduces the number of DNA probes, which is a factor hampering their miniaturization \cite{article_olgica_1,article_olgica_2}.

Our approach combines lessons from both fields but also has key differences from these frameworks.
The most obvious deviation of our
is that rather than focussing solely on maximal reduction in the number of {\it queries} (termed 'measurements' in
compressed sensing, or 'tests' in group testing) additional cost functions are introduced. 
Principally, we are interested in an additional objective of minimizing the weight of the design,
corresponding to the number of nonzero elements of the design matrix (for a fixed number of specimens).
This constraint originates from the properties of next generation sequencers, and prevents maximal query
reduction. We will discuss the consequences of this constraint and provide some theoretical bounds and
efficient designs. In particular, this theoretical framework is built on our recent experimental results regarding genotyping thousands of bacterial colonies using combinatorial pooling with NGSTs for a biotechnological application \cite{article_genome_research}. 
Prabhu et al.\cite{article_Itsik} develop a closely related theoretical approach to detect singletons using error correcting codes.
A somewhat similar compressed sensing approach has been independently developed by Shental et al.~\cite{article_Shental}.

The manuscript is divided as follows: In section \ref{section_setup}, we set up the basic formulation
of compressed genotyping. In section \ref{section_encoding}, we present the concept of light-weight
designs and provide a lower theoretical bound. Then, we show how constructions based on the Chinese 
Reminder Theorem comes close to this bound. In section \ref{section_decoding}, we 
present a Bayesian reconstruction approach based on belief propagation, and in section \ref{section_results}, we provide 
several simulations of carrier test, including Cystic Fibrosis. Section \ref{section_conclude} concludes the
manuscript.

\section{THE GENOTYPING PROBLEM - PRELIMINARIES}\label{section_setup}

\subsection{Notations}
We denote matices as an upper-case bold letter and the $(i,j)$ element of
the matrix $\mathbf{X}$ as $X_{ij}$. The shorthand $\overline{\mathbf{X}}$ denotes a matrix that its row
vectors are normalized and to sum to $1$. 
$\mathbb{I}(\mathbf{X})$ is an indicator function that returns a matrix in the same size as 
$\mathbf{X}$ with:
\begin{equation*}
	\mathbb{I}(X_{ij}) = 
 	\left\{ \begin{array}{rl}
			1 & X_{ij}>0 \\
			0 & X_{ij}=0 \\
		\end{array}\right.
\end{equation*}
For example: 
\begin{align*}
	\mathbf{X} & = \begin{bmatrix} 0 & 1 \\ 2 & 3\end{bmatrix} \\
	\overline{\mathbf{X}} & = \begin{bmatrix} 0 & 1 \\ 0.4 & 0.6\end{bmatrix} \\
	\mathbf{I}(\mathbf{X}) & = \begin{bmatrix} 0 & 1 \\ 1 & 1\end{bmatrix} \\
\end{align*}

The operation $|\cdot|$ denotes the cardinality of a set or the length of a vector.
For graphs, $\partial a$ refers to the subset of nodes that are connected to the node $a$, and 
the notation $\partial a \backslash b$ means the subset of nodes
that are connected to $a$ except of node $b$. 
We use natural logarithms.

\subsection{Genotyping As Bipartite Graph Reconstruction}
Consider a single human specimen that is being genotyped for a gene that has 
two alleles, labled by {\it A} and {\it a}.
We represent the genotype of this specimen by a vector of length two with three possible
outcomes: $(2,0)$ if the specimen is homozygous for the {\it A} allele, $(1,1)$ if the 
specimen is heterozygous, and $(0,2)$ if the specimen is homozygous for the {\it a} allele.
This representation can also accommodate situations where the gene has more than two alleles in 
the population by increasing the length of the vector to the number of the alleles.
The genotype of $n$ individuals is represented by an $n \times s$ matrix $\mathbf{G}$, called 
the {\it genotype matrix}, that is composed of the genotype vectors as its rows; $G_{ij}$ denotes 
how many copies of the $j$-allele the $i$-individual holds. For example, consider the follwing
genotyping matrix with $6$ individuals and only $2$ alleles:
\begin{equation*}
	\mathbf{G}=
	\left[ 
	\begin{array}{cc}
	2&0 \\
	2&0 \\
	2&0 \\
	1&1 \\
	2&0 \\
	0&2
	\end{array}				
 	\right] 
\end{equation*}
In that case, the $4^{th}$ individual is a carrier, the $6^{th}$ is affected, and the others are normal.

The genotyping matrix can be represented by a bipartite graph. Let $G$ be a bipartite multigraph,
$G=(X,S,E,p)$ that is built according to the genotype matrix,
where $X$ is a set of $n$ specimens, 
$S$ is a set of $s$ possible alleles in the population, and 
the edge set, $E$, denotes which subset of alleles each specimen holds. 
The degree of the specimen nodes,
$\deg(x_i)$, is a constant denoted by $p$, which represents the genome ploidy
\footnote{
	this assumption is valid 
	for the vast majority of the genotyping problems, but some particular cases that involve 
	copy number variation, such as Spinal Muscular Atrophy (SMA) \cite{artice_sma_is_copy_number},
	do not have a constant degree in their specimen nodes. They will remain outside the scope 
	of this manuscript. In addition, for the sex chromosomes in male $p=1$.
} 
 and in human $p=2$. The degree of each allele node, $\deg(s_i)$, is a random variable 
that is dictated by the prevalence of the genotypes in the population.
According to that representation, genotyping is in general the task of reconstructing 
the bipartite graph from the sequencing information - finding the edge set $E$ where $X$ and $S$ are 
known, subject to $p=2$. 

In a carrier screen, one is mainly interested in reconstructing a part of the graph, 
$E_{risk}$, that represents the subset of individuals that are carriers of the recessive risk alleles. 
The graph is very sparse due to the low prevalence of the risk alleles. 
Moreover, a large portion of severe genetic disease exhibit {\it complete penetrance}
 \cite{book_genes_to_genomes}, 
meaning that the affected individuals are symptomatic, 
and therefore are known, and do not participate in the screen. 
Thus, finding that one edge of a given individual is connected to a risk allele node 
immediately implies that the other edge is 
connected to a non-risk allele node, which further reduces 
the degrees of freedom in the graph reconstruction.

Consider two examples of carrier screens.
First, consider a screen for $\Delta F508$, the most prevalent mutation
 in Cystic Fibrosis (CF) among people of European descent (Fig. \ref{fig_bipart}).
In that case, the set $S$ has two members: WT and mutant, and the expectation of the ratio between 
$\deg(s_{mutant})$ to $\deg(s_{WT})$ is around 1:29
for screens in European \cite{article_cf_worldwide}\cite{article_cf_israeli}.
Most of specimen nodes sends double edges to the WT node, 
which implies that these specimens carry two copies of the
normal allele. A small portion of specimen nodes are connected to the two different 
allele nodes, meaning that these specimens are carriers for CF. There are no specimens 
that are connected by double edges to the mutant allele, since this mutation always 
leads to CF, and the affected individuals do not participate in the screen.
Consider also a screen with multiple risk alleles, as in the case of FMR1 gene that causes
Fragile X mental retardation\cite{article_FMR}. 
This gene has dozens of alleles, but only a small subset causes the disease. Therefore,
we need only to resolve edges to the risk alleles.
However, we are intrested in more than a binary classification of 
the specimens to carriers and normals, as the causative alleles carry different degrees
of disease risk (technically known as penetrance), and identifying the exact allele vector
has clinical utility.

\begin{figure}[!t]
\centering
\includegraphics[width=2in]{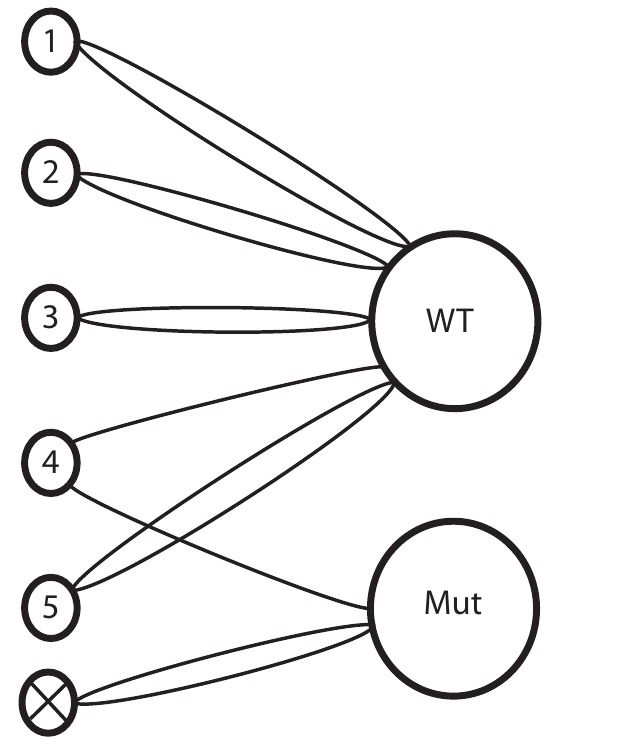}
\caption{Cystic Fibrosis $\Delta F508$ Screen as a Bipartite Multigraph Reconstruction.
There are two allele nodes, the WT and the $\Delta F508$ mutation.
Samples 1, 2, 3, 5 are WT, while specimen 4 is a carrier.
The specimen labeled with 'X' is affected and does not enter to the screen. $E_{risk}$ is the 
edge between specimen $4$ and the 'Mut' node.
}
\label{fig_bipart}
\end{figure}

\subsection{Defining a Cost-effective Reconstruction}
\begin{table*}[!t]
\caption{Summary of Query Design Parameters}
\centering
\begin{tabular}{p{1cm} p{5.5cm} p{2.5cm} p{6.8cm}}%
\hline\hline
\bf{Notation} & \bf{Meaning} & \bf{Typical Values} & \bf{Comments}  \\ [1ex]
\hline
$\mathbf{\Phi}$ & Query design & & \\[2ex]
$n$ & Number of specimens & Thousands & \\[2ex]
$t$ &  Number of queries (pools) & & \\[2ex]
$w$ & Weight & $\leqslant 8$  & Number of times a specimen is sampled \\[2ex]
$r_{max}$ & Max. level of compression & $\lesssim 1000$ & Maximal number of specimens in a pool \\[2ex]
$\tau_{max}$ & Number of queries in the largest query group & Up to a few hundreds & Corresponds to barcode synthesis reactions for a single experiment \\
\hline
\end{tabular}
\label{table_query_notation}
\end{table*}

Following the analysis above, a genotyping assay is a query of the form: ``which allele nodes
are connected to the interrogated specimen nodes?''. The current DNA sequencing 
methods that rely on serial specimen processing perform 
this query for each individual separately. However, the sparsity and the
restricted structure of $G$ suggest that $E_{risk}$ may be found in a relatively small number of queries when 
performing the queries on pools of specimens. Fortunately, the sequencing capacity of next 
generation sequencers enables the querying of pools of large numbers of specimens.

We will refer to the task of reconstructing $G$ in the most cost effective 
way as the {\it minimal genotyping problem}. Note that this task is intentionally 
not defined as minimizing the number of queries, since there are additional 
factors that determine the cost and the feasibility of the procedure.

We envision a minimal genotyping strategy that is based on a non-adaptive query schedule
in order minimize the turnover time and the need to re-pool the specimens multiple times. 
Our strategy starts by pooling samples of the specimens according to a certain design as
denoted by $\mathbf{\Phi}$, which is a $t \times n$ binary matrix; the columns of $\mathbf\Phi$ 
represent specimens, and each row determines a pool of specimens to be queried.
For example, if the first row of $\mathbf\Phi$ is $(1,0,1,0,1...)$, it specifies that the $1^{st}, 3^{rd}, 5^{th}, ...$ 
specimens are pooled and queried (sequenced) together.
Since the pooling is carried out using a liquid handling robot that can 
take several specimens in every batch, we only consider designs in which all specimens are 
sampled the same number of times to reduce the robotic logistics.
We define the {\it weight} $w$ of $\mathbf{\Phi}$ to be the number of times a specimen is sampled, or
the number of 1 entries in a given column vector: every column is constrained to 
have the same number of 1's in this design. Let $r_i$ be the {\it compression level} 
of the $i_{th}$ query, namely the number of 1 entries in the $i^{th}$ row, which 
denotes the number of specimens in the $i^{th}$ pool.

In large scale carrier screens, $n$ is typically between few thousands to tens of thousands of specimens.
We restrict ourself to query designs with $r_{max} \lesssim 1000$ specimens, due to technical 
/ biological limitations (in DNA extraction and PCR amplification) when processing pools  with larger
 number of specimens. 
A single query in such designs, even when composed of $1000$ specimens, 
does not saturate the sequencing capacity of next generation platforms. 
In order to fully exploit the capacity, 
we will pool queries together into {\it query groups} until the size of each group reaches the 
sequencing capacity limit, and we will sequence those groups in distinct reactions. 
Before pooling the queries, we will label each query with a unique DNA barcode 
in order to retain its identity 
(for an in-depth protocol of this approach see \cite{article_genome_research}).
Thus, the number of queries, $t$, is proportional to the number 
of DNA barcodes that should be synthesized, and one objective of the query design, similar
to those found in group testing and compressed sensing, is to minimize $t$.

In practice, once the DNA barcodes are synthesized, there is enough material for a few dozens 
experiments, and one can re-use the same barcode reagents for 
every query group as these are sequenced in distinct reactions.
Hence, the number of queries in the largest pooling group, $\tau_{max}$, dictates the synthesis cost
for a small series of experiments. While this does not change the asymptotic cost 
(e.g. for synthesizing barcodes for a large series of experiments),
it has some practical implications, and we will include it in our analysis.

The overall sequencing capacity needed for genotyping is proportional to $\sum_{i=1}^t{r_i}$, the 
total sum of the sizes of the queries. By definition 
$nw = \sum_{i=1}^t r_i$. Thus, the weight $w$ determines the requisite
sequencing capacity for a given number of specimens, and it is an additional factor that 
should be minimized in order to achieve a cost effective design. 
Moreover, decreasing the weight also reduces the number of times a specimen is sampled, 
and consequently, the robot time, and the amount of material that is consumed.
Notice that minimizing the weight of the design is not required by traditional compressed sensing 
construction.

Next generation sequencers are usually composed of several 
distinct biochemical chambers, called 'lanes', that can be processed in a single batch.
We assume that the sequencing capacity needed for $n$ specimens corresponds to one lane\footnote{
recent data have shown that when the number of specimens is a few thousand up to tens of thousands 
this assumption is valid \cite{article_nm_snps}}.
Since in total $nw$ aliquots of specimens are sampled in the pooling step, 
one needs $w$ lanes to sequence the entire design, where each lane is loaded with a different query group.

We do not intend to specify a global cost function that includes the costs of
barcode synthesis, robotic time, sequencing lanes, and other reagents. Clearly, 
these costs vary with different genotyping strategies, sequencing technologies, and so on. 
Rather, we will present heuristic rules that would be applicable in most situations.
First, $w$ should not exceed the maximal number of lanes
that can be processed in a single sequencing batch, as launching a run is expensive and time consuming,
and currently, for the most widespread next-generation sequencing platform, $w\leq8$ \cite{article_illumina_nature_paper}.
Therefore, it is also desirable that a design construction will have an 
explicit mechanism to specify the target weight.
Second, we assume that the cost of adding a sequencing lane is about two to three orders 
more than the cost of synthesizing an additional barcode. 
This will mainly served to benchmark the results of our design
to the outcome of other designs that were studied in group testing. Table \ref{table_query_notation} presents the notations we used in that part.

To conclude this part, the query design for the minimal genotyping problem is to
find a $t \times n$ design matrix composed of 0's and 1's $\mathbf{\Phi}$ that that provides sufficient
information to reconstruct $G$, while minimizing $t$ and keeping the column sum or weight $w$ below a
certain threshold. 
We term a design that addresses these objectives as {\it light weight design}.

\subsection {The Compositional Channel}
The sequencing procedure starts by capturing random DNA molecules from the input material,
and therefore, the ratios of sequence types reflect the corresponding ratios
of the genotypes in the input material. For example, consider a sample that is composed of a mixture of two
specimens in equal ratio, where one specimen is homozygous WT and the other is heterozygous. 
About $3/4$ of the sequence reads will correspond to the WT 
allele and $1/4$ to the mutant genotype.  
Since the input material in our case is composed of pools of specimens, the sequencing results 
are given by the following conditional probability:
 \begin{equation}\label{equ_compositional_channel_no_errors}
         f_{\beta}(\mathbf{Y}\mid\mathbf{\Phi G})
\end{equation}
where $\mathbf{Y}$ is a $t \times s$ matrix that denotes the sequencing results, namely the number of 
sequence reads for each genotype/query combination, 
and $\mathbf{G}$ is the $n \times s$ biadjacency matrix 
of the genotyping graph. $\beta$ is a sampling parameter, a non-negative integer 
that denotes the number of reads for each query. In reality, $\beta$ is a random variable
with Poisson distribution, since each query has different number of reads. However,
for simplicity we will treat it as a constant.
$f_{\beta}(\mathbf{Y}\mid\mathbf{X})$ denotes a multinomial random 
process that corresponds to the sampling procedure of the 
sequencers, the joint distribution of the sequencing results is therefore given by 

\begin{equation}
	f_{\beta}(\mathbf{Y}\mid\mathbf{X}) = 
	{\displaystyle \prod_{i=1}^t}
		\alpha_{i}
       		\exp
       		\left( 
		-\beta \sum_{j=1}^{s}
				\overline{Y}_{ij}
				\log(
					1/\overline{X}_{ij}
				      )
       		\right)
\end{equation}
where $\alpha_{i}=\frac{\beta!}{\prod_{j=1}^{s}Y_{ij}!}$. As 
$\beta\rightarrow\infty$ the relative ratios of a row in $\mathbf{Y}$ become similar to 
the ratios of allele nodes degrees of the subgraph induced by the specimens in the pool.
We will term the process in Eq. (\ref{equ_compositional_channel_no_errors}) 
{\it compositional channel} with parameter $\beta$.
The reason that we used this name is that the channel places $s$-dimensional real
space input vectors in an $s-1$-dimensional simplex, which is reminiscent of the concept 
of compositions in data analysis \cite{book_compositional}.

The compositional channel is closely related to two other channels, the {\it superimposed channel}, and
the {\it real adder channel}. As we mentioned earlier, the superimposed channel has been 
extensively studied in the group testing literature, and describes queries 
that only return the presence or absence of the tested feature among 
the members in the pool. On pooled data, measured without noise, the superimposed channel 
would be given by:
 \begin{equation}\label{equ_superimposed_channel}
        \mathbf{Y}_{s}= \mathbb{I}(\mathbf{\Phi G})
\end{equation}
The information degradation here is more severe than in the case of the compositional channel,
since the observer can not quantify the number of positive items from a single query with positive
answer. 
The output of the compositional channel can be further processed as if it was obtained 
by a superimposed channel.
In that case $\mathbf{Y}_{s}$ denotes only the presence / absence of 
an allele in a query.
This degradation is given by:
\begin{equation}\label{superimposed_degredation}
	\mathbf{Y}_{s} = \mathbb{I}(\mathbf{Y})
\end{equation}

The real adder channel describes the result of a query as a linear combination of the samples 
in the pool, and is given by:
\begin{equation}\label{linear_channel}
	\mathbf{Y}_{l} = \mathbf{\Phi G}
\end{equation}
This type of channel serves as the main model for compressed sensing, and it captures many physical
phenomena.
A closely related models were studied in group testing for finding
counterfeit coins with a precise spring scale\cite{article_coin_problem} and in multi-access
 communication
\cite{article_channel_reviews}\cite{article_t_adder_channel_1}\cite{article_t_adder_channel_2}. 
Ideally, when one knows the number of specimens in each pool, and $\beta \rightarrow\infty$,
data from the compositional channel can be treated as if it was obtained by a real adder channel, since
the compositional vectors can be placed back in the real space by normalizing 
$\mathbf{Y}$ to $\overline{\mathbf{Y}}$ and multiplying the result with 
the number of specimens in each query.

In reality, the sequencer may also produce errors when reading 
the DNA fragments. Since the fragments are composed of a barcode region 
and the interrogated region, sequencing errors may lead to association of sequence reads
with the wrong query, and to an incorrect genotype detection.
DNA barcodes can be easily extended in order to add more redundancy to the codeword that they carry, 
and experience has demonstrated that 
that errors in barcode annotation are insignificant\cite{article_hamady}, and we will not treat
this type of errors. On the other hand, sequencing errors in the interrogated region 
are more involved and we will classify them into two categories according to their outcome: 
{\it confounding errors}, meaning that a sequence read that was derived from one genotype 
is decoded as another valid genotype,
and {\it non-sense errors}, meaning that a decoded sequence read does not correspond to any allele node
in $G$. Non-sense errors are easily handled, for instance, by filtering those sequence reads, as we assume
 that all possible alleles are known beforehand. Unfortunately, 
there is no simple remedy for confounding errors, and they may distort the 
observed allele ratios in the queries. For instance, consider a pool of 100 specimens none of 
which has a mutant allele that is sampled with $\beta\rightarrow\infty$. If the confounding error rate is 
$0.5\%$, the data will falsely indicate that one of the specimens in the pool is a carrier.
The effect of the sequencing errors on genotyping may be denoted by a conditional probability which includes a confusion matrix:
\begin{equation}\label{equ_compositional_channel_with_errors}
         f_{\beta}(\mathbf{Y}\mid\mathbf{\Phi G \Lambda})
\end{equation}
$\mathbf{\Lambda}$ denotes an $s\times (s+1)$ confusion matrix that
indicates the probability that the $i^{th}$ genotype is confounded with the $j^{th}$ genotype; the $s+1$
column indicates the transition probability to a non-sense genotype. We will call the process in Eq.
(\ref{equ_compositional_channel_with_errors}) {\it compositional channel with errors}. 

The values of $\mathbf{\Lambda}$ are dependent on the sequence differences between the
interrogated alleles and on the specific chemistry that is utilized by the sequencing platform. Based on
previous work regarding the most abundant next generation sequencer \cite{article_nm_snps}\cite{article_me_nm}, we presume 
that subtle mutation differences (known as SNPs), such as $W1282X$ mutation in CF or 
$Y231X$ mutation in Canavan disease, correspond to confounding error rates of up to $1\%$,
and when the sequence differences are more profound, like in $\Delta F508$ mutation in CF, 
we expect that the $s \times s$ left submatrix in $\mathbf{\Lambda}$ will resemble the identity matrix.
Table \ref{table_channels} summarizes the different channel models in this section.

\begin{table}[!t]
\caption{Comparison of Different Channels Models}
\centering
\begin{tabular}{ l  l  l }%
\hline\hline
\bf{Channel model} & \bf{Measurement process}  & \bf{Example}  \\ [1ex]
\hline
Superimposition & OR operation              & Antibody reactivity                   \\[1ex]
Compositional    & Multinomial sampling & Next generation sequencing  \\[1ex]
Real Adder                  & Additive                        & Spring scale                             \\ [1ex]
\hline
\end{tabular}
\label{table_channels}
\end{table}

\section{QUERY DESIGN}\label{section_encoding}
\subsection{Constraints On Light-Weight Designs}
Group testing theory suggests a sufficient condition for $\mathbf{\Phi}$, called d-disjunction, 
that ensures faithful and tractable reconstruction of any up to $d$ sparse vector that 
was obtained from a superimposed channel\cite{article_kautz_singelton}. 
Since the compositional channel can be degraded to a
 superimposed channel, d-disjunction is also a sufficient criterion for reconstructing a 
$d$ sparse vector over a noise free compositional channel.
Respecting a carrier screen and reconstruction $E_{risk}$, if each risk allele node has 
less than $d$ edges, d-disjunction is a sufficient condition to reconstruct $E_{risk}$ 
given a sufficient sequencing depth and no errors.

Based on the analysis about cost effective designs, we are looking for non-trivial d-disjunct matrices
that reduce the number of barcodes with a minimal increase of the weight.

\begin{definition}\label{def_disjunct}
$\mathbf{\Phi}$ is called {\it d-disjunct} if and only if the Boolean sum of up to 
$d$ column vectors of the matrix does not include any other column vector.
\end{definition}

\begin{definition}\label{def_reasonble}
$\mathbf{\Phi}$ is called {\it reasonable} if it does not  contain a row with only a 
single 1 entry, and its weight is more than $0$.
\end{definition}
We are only interested in reasonable designs. Clearly, if a design includes queries 
composed of single specimens, it is more cost effective to genotype those specimens in serial processing.

\begin{definition}\label{def_lambda}
	$\lambda_{ij}$ is the dot-product of two column vectors of $\mathbf{\Phi}$, and 
	$\lambda_{\max} \triangleq \max(\lambda_{ij})$.
\end{definition}

\begin{lemma}\label{lem_minweight}
	The minimal weight of a reasonable d-disjunct matrix is: $w=d+1$. 
\end{lemma}
\begin{proof}
	Assume that $\lambda_{max}$ occurs between $\overrightarrow {C(i)}$ and $\overrightarrow {C(j)}$,
	 two column vectors in $\mathbf{\Phi}$.
	According to definition (\ref{def_reasonble}), every 1 entry in $\overrightarrow {C(i)}$ 
	intersects with at least one column vector. Thus, there are at most $w$ column vectors 
	that intersect with 	$\overrightarrow {C(i)}$. The Boolean sum of those $w$ column 
	vectors includes 	$\overrightarrow {C(i)}$, so the matrix is not w-disjunct. 
	According to definition (\ref{def_disjunct}), it can not be d-disjunct, and $w \geq d+1$. 
	The existence d-disjunct matrices with $w = d+1$ was proved by Kautz and Singleton\cite{article_kautz_singelton}.
 \end{proof}

\begin{definition}\label{def_minweight}
	$\mathbf{\Phi}$ is called {\it light-weight} d-disjunct in case $w=d+1$. 
\end{definition}

\begin{lemma}\label{lem_lambda_eq_one}
 	$\mathbf{\Phi}$ is a light-weight $(w-1)$-disjunct iff 
	$\lambda_{max} =1$ and $\mathbf{\Phi}$ is reasonable.
\end{lemma}
\begin{proof}
	First we prove that if $\mathbf{\Phi}$ is a light-weight $(w-1)$-disjunct then 
	$\lambda_{max} =1$.
	Assume $\lambda_{max}$ occurs between $\overrightarrow{C(i)}$ and
	$\overrightarrow{C(j)}$. According to definition (\ref{def_reasonble}), 
	there is a subset of at most $w-\lambda_{max} +1$ column vectors
	that $\overrightarrow{C(i)}$ is included in their Boolean sum.
	However, the Boolean sum of any $w-1$ column vectors does not 
	include $\overrightarrow{C(i)}$. Thus, $\lambda_{max}<2$, and according 
	to definition (\ref{def_reasonble}), $\lambda_{max} > 0$. 
	Thus, $\lambda_{max} =1$. In the other direction, Kautz and Singelton\cite{article_kautz_singelton} 
	proved that $d=\lfloor{(w-1)}/\lambda_{max}\rfloor$, and $\mathbf{\Phi}$ is light-weight
	according to definition (\ref{def_minweight}).
\end{proof}

\begin{lemma}\label{lem_bound_of_singelton}
	The number of columns of $\mathbf{\Phi}$ is bounded by: 
	\begin{equation}
		n \leq \frac
		{\left(
				\begin{array}{c}
				t	\\	\lambda_{max}+1
				\end{array}
		\right)}
		{\left(
				\begin{array}{c}
				w	\\	\lambda_{max}+1
				\end{array}
		\right)}
	\end{equation}
\end{lemma}
\begin{proof}
	see Kautz and Singelton \cite{article_kautz_singelton}.
\end{proof}

\begin{theorem}\label{thr_lower_bound}
	The minimal number of rows, $t$, in a light-weight d-disjunct matrix is $t>\sqrt{w(w-1)n}$ 
\end{theorem}
\begin{proof}
	plug lemma (\ref{lem_lambda_eq_one}) to lemma (\ref{lem_bound_of_singelton}): 
	\begin{align*}
		n &\leq \frac{\left(\begin{array}{c}t\\2\end{array}\right)}{\left(\begin{array}{c}w\\2\end{array}\right)} \\
		n  &\leq \frac{t^2}{w(w-1)} \\ %\frac{t(t-1)}{w(w-1)}<
	\sqrt{w(w-1)n} &\leq t
	\end{align*}
\end{proof}

\begin{corollary}\label{cor_min_tao}
	The minimal number of barcodes is $\tau_{max} \simeq \sqrt{n}$
	in a light-weight weight design.
\end{corollary}
\begin{proof}
	There are $w$ query groups, and the bound is immediately derived from theorem (\ref{thr_lower_bound}).
\end{proof}
A light weight design is characterized by $\lambda_{max}=1$, and $t\sim \Omega((d+1)\sqrt{n})$
rows. The low $\lambda_{max}$ attribute does not only increase the disjunction of the matrix but also 
eliminates any short cycle of 4 in the factor graph built upon $\mathbf\Phi$, which enhances 
the convergence of the reconstruction algorithm. We will discuss this
property in section \ref{section_decoding}.

\subsection{Light Chinese Design}
We suggest a light weight design construction based on the Chinese Remainder Theorem.
This construction reduces the number of queries to the vicinity of 
the lower bound derived in the previous section, 
and can be tuned to different weights and numbers of specimens. The repetitive structure of the 
design simplifies its translation to robotic instructions, and permits easy monitoring.

Constructing $\mathbf{\Phi}$ starts by specifying: (a) the number of specimens, and (b) the required disjunction, which immediately determines the weight. Accordingly, a set of $w$ positive integers 
$Q=\{q_1, \dots ,q_w\}$, called {\it query windows}, is chosen with the following requirement: 
\begin{equation}\label{equ_assert}
	\forall \{q_i,q_j\}, { i \neq j} :	\;
	lcm(q_i,q_j) \geq n
\end{equation}
where $lcm$ denotes the least common multiplier.
We map every specimen $x$ to a residue system $(r_1,r_2,\cdots,r_w)$ according to:
\begin{equation}\label{equ_pooling_rule}
	\begin{array}{c}
 	x\equiv r_1 \pmod {q_1} \\
	x\equiv r_2 \pmod {q_2} \\
	\vdots \\
	x\equiv r_w \pmod {q_w} \\
	\end{array}
\end{equation}
Then, we create a set of $w$ all-zero sub-matrices $\mathbf{\Phi}^{(1)},\mathbf{\Phi}^{(2)},\cdots$ 
called {\it query groups} with sizes $q_i \times n$. The submatrices captures the mapping in 
Eq. (\ref{equ_pooling_rule}) by setting $\Phi_{rx}^{(i)}=1$ when  this clause: $x\equiv r \pmod{q_{i}}$
is true. Finally, we vertically concatenate
the submatrices to create $\mathbf{\Phi}$:
\begin{equation}\label{equ_concatanate}
	\mathbf{\Phi}=
	\left[ 
	\begin{array}{ccccc}
		& & [\mathbf{\Phi_1}]  & & \\ \hline
		& & [\mathbf{\Phi_2}] &  & \\ \hline
		& & \vdots & &\\ \hline		
		& & [\mathbf{\Phi_w}] & &	
	\end{array}				
 	\right] 
\end{equation}
For instance, this is\footnote{
	when $x=q_i$ we set $r_i=q_i$, so the first row in every submatrix is $1$
}
$\mathbf\Phi$ for $n=9$, and $w=2$, with $\{q_1=3,q_2=4\}$:
\begin{equation*}
	\mathbf{\Phi}=
	\left[ 
	\begin{array}{ccccccccc}
	1&0&0&1&0&0&1&0&0 \\
	0&1&0&0&1&0&0&1&0 \\
	0&0&1&0&0&1&0&0&1 \\ \hline
	1&0&0&0&1&0&0&0&1 \\
	0&1&0&0&0&1&0&0&0 \\
	0&0&1&0&0&0&1&0&0 \\
	0&0&0&1&0&0&0&1&0
	\end{array}				
 	\right] 
\end{equation*}

\begin{definition}\label{def_minimal_chinese_design}
	Construction of $\mathbf{\Phi}$ according to 
	Eq. (
		\ref{equ_assert},
		\ref{equ_pooling_rule},
		\ref{equ_concatanate}
	) is called light Chinese design.
\end{definition}

\begin{theorem}\label{thr_d_crt_dijunct}
	A light Chinese design is a light weight design.
\end{theorem}
\begin{proof}
	Let $x \equiv v_x\pmod {q_i}$ and $x \equiv u_x\pmod {q_j}$, where $i \neq j$. 
	According to the Chinese Remainder Theorem there is a one-to-one correspondence
	$\forall x: x \leftrightarrow (u_x, v_x)$. Thus, every two positive entries in 
	$\overrightarrow{C_x}$ are unique. 
	Consequently, $|\overrightarrow{C_x} \cap \overrightarrow{C_y}|<2$, and $\lambda_{max}<2$.
	Since specimens in the form $r, r+q_i, r+2q_i, \dots$ are pooled together $\mathbf{\Phi}$,
	$\lambda_{max}=1$. 
	According to lemma (\ref{lem_lambda_eq_one}) $\mathbf{\Phi}$ is light-weight.
\end{proof}

\subsection{Choosing the Query Windows}
The set of query windows, $Q$, determines the number of rows in $\mathbf{\Phi}$
as:
\begin{equation}\label{equ_number_of_rows}
	t=\sum_{i=1}^w{q_i}
\end{equation}
Since $lcm(x,y)=\frac{xy}{gcd(x,y)}$, where $gcd$ is the greatest common divisor, minimizing 
the elements in $Q$ subject to the constraint in Eq. (\ref{equ_assert}) 
implies that $q_1,\dots, q_w$ should 
be pairwise coprimes and $q_i \geq \sqrt{n}$.
Let $\kappa = \lceil \sqrt{n} \rceil -1$, the definition of the problem we
seek to solve is as follows: given a threshold, $\kappa$, and $w$, a valid
solution is a set, $R$, that contains $w$ co-prime numbers, all of which are
larger than $\kappa$. We seek for the optimal solution, $Q$, being the solution
satisfying that $\sum(Q)$ is minimal.

We begin by introducing a bound on $\max(Q)-\kappa$, a value we will name $\delta$, or 
the discrepancy of the optimal solution. In order to give an upper bound on $\delta$, let us first consider
a bound that is not tight, $\delta_0$, the discrepancy of the solution $Q_0$ that is composed
of the $w$ smallest primes greater than $\kappa$. 
Primes near $\kappa$ have a density of $1/\log(\kappa)$, so $\delta_0\approx w\log(\kappa)$.
$\delta_0$ is known to be an upper bound on $\delta$
because if any value greater than
$\kappa+\delta_0$ appears in the $Q$, then there is also a prime $q$,
$\kappa<q\le \max(Q_0)$ that is not
used. There is at most one value in $Q$ that is not co-prime with $q$ and if it
exists it is larger than $q$. Replace it by $q$ in $Q$ (or replace $\max(Q)$
with $q$ if all numbers in $Q$ are co-prime with $q$) to reach a better
solution, contradicting our assumption that $Q$ is the optimal solution.

This upper bound can improved as follows. We know that $Q \subset (\kappa,\kappa+\delta_0]$. In
this interval, there is at most one value that divides any number greater or
equal to $\delta_0$. Consider, therefore, the solution $Q_1$ composed of the $w$
smallest numbers larger than $\kappa$ that have no factors smaller than $\delta_0$. In
order to assess the discrepancy of this solution, $\delta_1$, note that the
density of numbers with no factors smaller than $\delta_0$ is at least $1/\log(\delta_0)$.
This can be shown by considering the (lower) density that is the density of
the numbers with no factors smaller than $p_{\delta_0}$, where $p_i$ indicates the
$i$'th smallest prime. This density is given by:
\begin{align}
	\prod_{i<\delta_0} 1-\frac{1}{p_i}
	& = e^{\log\left(\prod_{i<\delta_0} 1-\frac{1}{p_i}\right)} \nonumber \\
	& =e^{\sum_{i<\delta_0} \log(1-\frac{1}{p_i})} \nonumber \\ 
	& \approx e^{\sum_{i<\delta_0} -\frac{1}{p_i}} \\
	& \approx e^{-\log(\log(\delta_0))} \nonumber \\
	&= \frac{1}{\log(\delta_0)}  \nonumber \\\nonumber
\end{align}
where $e$ is Euler's constant and we make use of $\sum_{i<\delta_0} \frac{1}{p_i}
\approx \log(\log(\delta_0))$, a well-known property of the prime harmonic series.
Like $\delta_0$, the bound $\delta_1$ is also an upper bound on $\delta$.
To show this, consider that the optimal solution $Q$ may have $z$ values larger
than $\kappa+\delta_1$ in it. If so, there are at least $z$ members of $Q_1$ absent from
it. Replace the $z$ members of $Q$ with the absent members of $Q_1$ to reach
an improved solution. We conclude that $z=0$ and $\delta_1 \ge \delta$.

\begin{theorem}
	For $\kappa \rightarrow \infty$ and large $w$, $\delta \approx w \log(w)$.
\end{theorem}
\begin{proof}
	Consider repeating a similar improvement procedure as was 
	used to improve from $\delta_0$ to $\delta_1$ 
	an arbitrary number of times. We define $Q_{i+1}$ as the set of $w$ minimal 
	numbers that are greater than $\kappa$  and have no factors smaller than $\delta_i$, 
	where $\delta_i$ is the discrepancy of solution $Q_i$. This 
	creates a series of upper bounds for $\delta$ that is monotone 
	decreasing, and therefore converges. Because each $\delta_{i+1}$ satisfies
	$\delta_{i+1} \approx w \log(\delta_i)$, we conclude that the limit will satisfy
 	$\delta_\infty \approx w \log(\delta_\infty)$, meaning $\delta \le \delta_\infty \approx w \log(w)$.
	This gives an upper bound on $\delta$.
	To prove that this bound is tight, we will show that, asymptotically, it is
	not possible to fit $w$ co-prime numbers on an interval of size less than
	$w \log(w)$. To do this, note first that at most one number in the set can be
	even. Fitting $w-1$ odd numbers requires an interval of size at least $2w$
	(up to a constant). The remaining numbers can contain at most one value that
	divides by $3$. The rest must be either $1$ or $2$ modulo $3$. This indicates
	that they require an interval of at least $2 \cdot \frac{3}{2} w$. More
	generally, if $S$ contains $w$ values, with each of the first $w$ prime numbers
	dividing at most one of said values, then the interval length of $S$ must be
	at least on the order of:
	\begin{align*}
			w \prod_{i<w} 1+\frac{1}{p_i} 
			&= w e^{\log\left(\prod_{i<w} 1+\frac{1}{p_i}\right)} \\
			&= w e^{\sum_{i<w} \log\left(1+\frac{1}{p_i}\right)} \\
			& \approx w e^{\sum_{i<w} \frac{1}{p_i}} \\
			& \approx w e^{\log(\log(w))} \\
			& = w \log(w)\\
	\end{align*}
	This gives a lower bound on $\delta$ equal to the previously calculated upper
	bound, meaning that both bounds are tight.
\end{proof}

\begin{corollary}
	$\tau_{max} = \sqrt{n}+w \log(w)$
\end{corollary}
\begin{proof}
	\begin{align*}
	 \max(Q) - \kappa &= \delta  \\
	\max(Q)  &= \sqrt{n} + \delta \\
	 \max(Q)  &= \sqrt{n} + w log(w) \\
	\end{align*}
	Since the number of positive entries in each submatrix is the same and equals to $n$
	the query groups are formed by partitioning $\mathbf{\Phi}$ to the submatrices. Consequently,
	$\tau_{max} = \max(Q)$.
\end{proof}

Importantly, the maximal compression level, $r_{max}$, is never more than $\sqrt{n}$, and the
light Chinese design is practical for genotyping tens of thousands of specimens.
The tight bound on $\delta$ also implies a tight bound on the sum of $Q$. 
Let $\sigma_{Q} = \sum_{i=1}^{w}{q_i} - w\kappa$.  We give a tight bound on
$\sigma_{Q} - w\kappa$  that, asymptotically, reaches a 1:1 ratio with the optimal value.

\begin{theorem}\label{thr_lower_crt}
	The number of queries in the light Chinese design is
	$t \approx \Theta(w\kappa+\frac{1}{2} w^2 \log(w))$
\end{theorem}
\begin{proof}
 	Proof that this is a lower bound is	
	by induction on $w$. Specifically, let us suppose the claim is true for $Q_{w-1}$
	and prove for $Q_w$. (There is no need to verify the ``start'' of the
	induction, as any bounded value of $w$ can be said to satisfy the approximation
	up to an additive error.) To prove a lower bound, $\sigma_{Q_w}$
	cannot be better than $\sigma_{Q_{w-1}} + w \log(w)$, as the discrepancy of $Q_w$
	is known and the partial solution $Q_w \setminus \max(Q_w)$ can not be better
	than $Q_{w-1}$.
	To prove that this is also an upper bound, consider that
	the discrepancy of $Q_w$ is known to be approximately $w \log(w)$, so any prime
	larger than approximately $p_w$ cannot be a
	factor of more than one member of the interval $(\kappa,\max(Q_w)]$.
	Furthermore, the optimal solution for $Q_w$ can not be significantly worse than
	the optimal solution for $Q_{w-1}$ plus the first number that is greater than
	$\max(Q_{w-1})$ and has no factors smaller than $p_w$. As we have shown before,
	this number is approximately $\max(Q_{w-1})+\log(w)$. However, we already know
	the discrepancy of $Q_{w-1}$ is approximately $(w-1) \log(w-1)$, so this new
	value is approximately $\kappa+(w-1) \log(w-1)+\log(w) \approx \kappa+w \log(w)$. Putting
	everything together, we get that
	$\sigma_{Q_w} \le \sigma_{Q_{w-1}} + w \log(w) \le \sum_{i \le w} i \log(i)$,
	proving the upper bound.	
	The value $\sum_{i \le w} i \log(i)$ is between $\frac{1}{2} w^2 \log(w-1)$
	and $\frac{1}{2} w^2 \log(w)$, so asymptotically $\sigma_Q$ converges to
	$\frac{1}{2} w^2 \log(w)$.
\end{proof}

We will now consider algorithms to actually find $Q$. 
First, consider an algorithm that begins by setting $\tau_{max}$ to the $w$ prime number 
after $\kappa$, and then runs an exhaustive search through all sets of
size $w$ that contain values between $\kappa$ and $\tau_{max}$. 
This is guaranteed to return the optimal result, and does so in complexity 
$O((\tau-\kappa)^w)$, which is asymptotically equal to 
$O((w \log(\kappa))^w)$. Though this complexity is hyper exponential, 
and so unsuitable for large values of $w$ it may be used for
smaller $w$. 

The upper bound described above suggests a polynomial algorithm for $Q$
since it is a bound that utilizes sets chosen such that
none of their elements have prime factors smaller than $\kappa$.
This implies the following simplistic algorithm that calculates a solution that
is asymptotically guaranteed to have a 1:1 ratio with the optimal $\sigma_Q$.

\begin{algorithmic}[1]
\STATE Let $Q$ be the set of the $w$ smallest primes greater than $\kappa$.
\REPEAT
\STATE $\delta \gets \max(Q)-\kappa$
\STATE $Q \gets$ the $w$ smallest numbers greater than $\kappa$ that have no
factors smaller than $\delta$
\UNTIL {$\delta=\max(Q)-\kappa$}
\STATE output $Q$.
\end{algorithmic}

In practice, this is never the optimal solution, as for example, it contains
no even numbers. In order to increase the probability that we reach the
optimal solution (or almost the optimal solution), we opt for a greedy version
of this algorithm. The greedy algorithm begins by producing the set of smallest
numbers greater than $\kappa$ that have no factors smaller than $\delta$ (as in the upper
bound). It continues by producing the set of smallest co-prime numbers greater
than $\kappa$ that have at most one distinct factor smaller than $\delta$ (as in the
calculation of the lower bound). Then, it attempts to add further elements with
a gradually increasing number of factors. If these attempts cause a decrease in
$\delta$, it repeats the process with a lower value of $\delta$ until reaching
stabilization.

\begin{algorithmic}[1]
\STATE $Q \gets$ initial solution.
\REPEAT
\STATE $\delta \gets \max(Q)-\kappa$
\STATE $n(x) \defeq$ the number of distinct primes smaller than $\delta$ in the
factorization of $x$.
\STATE Sort the numbers $\kappa+1, \ldots, \max(Q)$ by increasing $n(x)$ [major key]
and increasing value [minor key].
\FORALL {$i$ in the sorted list}
\IF {$i$ is co-prime to all members of $Q$ and $i<\max(Q)$}
\STATE replace $\max(Q)$ by $i$ in $Q$.
\ELSIF {$i$ is co-prime to all members of $Q$ except one, $q$, and $i<q$}
\STATE replace $q$ with $i$ in $Q$.
\ENDIF
\ENDFOR
\UNTIL {$\delta=\max(Q)-\kappa$}
\STATE output $Q$.
\end{algorithmic}

Because this greedy algorithm only improves the solution from iteration to
iteration, using the output of the first algorithm described as the initial
solution for it guarantees that the output will have all asymptotic optimality
properties proved above. In practice, on the range $\kappa=100 \ldots 299$ and
$w=2 \dots 8$ it gives the exact optimal answer in 91\% of the cases
and an answer that is off by at most $2$ in 96\% of the cases.
(Understandably, no answer is off by exactly $1$.) The worst results for it
appear in $w=8$, where only 82\% of the cases were optimal and
88\% of the cases were off by at most $2$.

Notably, due to the fact that $\kappa+1$ does not always appear in either the optimal
solution or the solution returned by the greedy algorithm, sub-optimal results
tend to appear in {\it streaks}: a sub-optimal result on a particular $\kappa$ value
increases the probability of a sub-optimal result on $\kappa+1$ (A similar
property also appears when increasing $w$), and we denote an interval of consecutive
$\kappa$ values where the greedy algorithm returns sub-optimal results to be a
``streak''. The number of streaks is, perhaps, a better indication for the
quality of the algorithm than the total number of errors. For the parameter
range tested (totaling 1400 cases), the greedy algorithm produced $61$
sub-optimal streaks (of which in only $23$ streaks the divergence from the
optimal was by more than 2). The worst $w$ was 6, measuring 14 streaks. The
worst-case for divergence by more than 2 was $w=8$, with 8 streaks.

In terms of the time complexity of this solution, this can be bounded as
follows. First, we assume that the values in the relevant range have been
factored in advance, so this does not contribute to the running time of the
algorithm. (This factorization is independent of $\kappa$ and $w$, except in the
very weak sense that $\kappa$ and $w$ determine what the ``relevant'' range to
factor is.) Next, we note that the initial $\delta$ is determined by searching for
$w$ primes, so we begin with a $\delta$ value on the order of $w \ln(\kappa)$. Each
iteration decreases $\delta$, so there are at most $w \log(\kappa)$ iterations. 
In each iteration, the majority of time is spent on sorting $\delta$ numbers. Hence, the
running time of the algorithm is bounded by $\delta^2 \log(\delta)$ or
$w^2 \log^2(\delta) (\log(w) + \log(\log(\kappa)))$. Clearly, this is a polynomial solution.
In practice, it converges in only a few iterations, not requiring the full
$\delta$ potential iterations. In fact, in the tested parameter range the algorithm
never required more than three iterations in any loop, and usually less. (Two
iterations in the greedy allocation loop is the minimum possible, and an
extra iteration over that was required in only $4\%$ of the cases.)

In some cases, it is beneficial to increase the number of barcodes from $\kappa$ to $\kappa_1$ in order to 
achieve higher probability of faithful reconstruction of signals that are not $d$ sparse. This is achieved
by finding $w$ integers in the interval $(\kappa,\kappa_1)$ that follow Eq. (\ref{equ_assert}) and maximize
$\prod_{i=1}^{w}{q_i}$. We give more details about this problem in appendix \ref{app_product_max}.

Lastly, we note that even though these approximation algorithms are necessary
for large values of $w$, for small $w$ an exhaustive search for the exact
optimal solution is not prohibitive, even though the complexity of such a
solution is exponential. One can denote the solution as
$Q=\{\kappa+s_1,\kappa+s_2, \ldots, \kappa+s_w\}$ in which case the values
$\{s_1, \ldots, s_w\}$ are only dependent on the value of $\kappa$ modulo
primes that are smaller than the maximal $\delta$ or approximately
$w \log(w)$. This means that the $\{s_1, \ldots, s_w\}$ values in the optimal
solution for any $(\kappa,w)$ pair is equal to their values for
$(\kappa \mod P,w)$, where $P$ is the product of all primes smaller than
$\delta$. Essentially, there are only $P$ potential values of $\kappa$ that
need to be considered. All others are equivalent to them.

In practice, the number of different $\kappa$ values that need to be considered
is significantly smaller than this. As an example, Fig. \ref{fig_tree} gives the
complete optimal solution for any value of $\kappa$ with $w=4$. The figure
shows that the set $S=\{s_1,\ldots,s_w\}$ for any $\kappa$ has only 7
possible values, and that determining which set produces the optimal
solution for any particular value of $\kappa$ can be done by at most 5 Boolean
queries regarding the value of $\kappa$ modulo specific primes.

\begin{figure}[!t]
\centering
\includegraphics[width=2.5in]{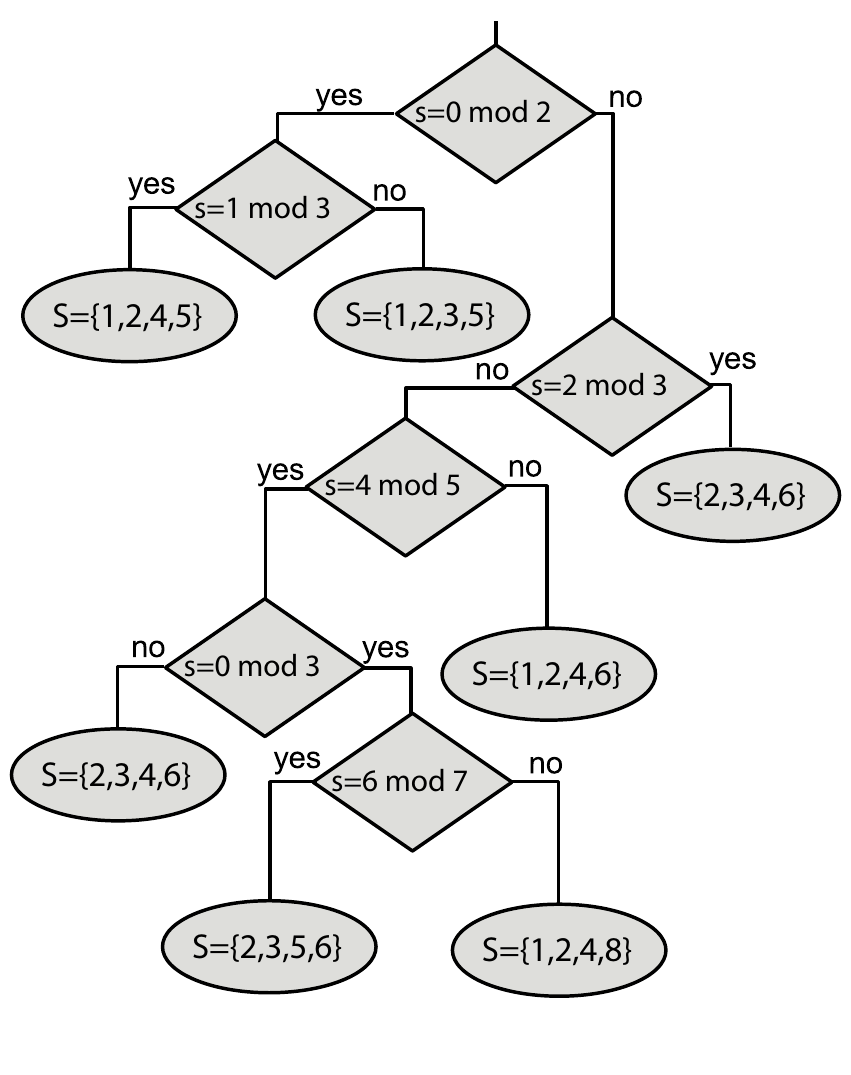}
\caption{Optimal Solution using Tree Search}
\label{fig_tree}
\end{figure}

\subsection{Comparison to Logarithmic Designs}
It is well established in group testing theory and in compressed sensing that 
certain designs can reach to the vicinity of the lower theoretical bound of $t \sim O(d\log{n})$ \cite{article_Dyachkov,article_measurements_vs_bits}. 
The $t \sim O((d+1)\sqrt{n})$ scale in light weight designs raises the question whether they
are really the most cost effective solution for the minimal genotyping problem with thousands of specimens
and $w \leqslant 8$.

We compared the results of the light Chinese design with the method of Eppstein and
 colleagues\cite{article_epsstein} for screens with 5000 and 40000 specimens (Table \ref{table_eppstein}). 
To the best of our knowledge, Eppstein's method shows for the general case the maximal 
reduction of $t$. Interestingly, it is also based on the Chinese Remainder Theorem,
but without the assertion in Eq. (\ref{equ_assert}). Instead, their method requires that $Q$ will be 
composed of co-prime numbers whose product is more than $n^d$ in order to create a d-disjunct matrix.
The number of queries in their method for a given $d$ is:
\begin{equation}\label{equ_eppstein_qureis}
	t \sim O(d^2 \log^2{n}/(\log{d} + \log{\log{n}}))
\end{equation}
and the weight is:
\begin{equation}\label{equ_eppstein_weight}
	w \sim O(d \log n /  (\log{d} + \log{\log{n}}))
\end{equation}
Notice that their weight scales with the number of specimens, implying that more sequencing lanes
 and robotic logistic are required with the growth of $n$ even if $d$ is constant.

First, we found that Eppstein's method is not applicable to the biological 
and technical constraints in the genotyping setting of $w\leqslant 8$ and $r_{max} \lesssim 1000$ 
(labeled in the table with $\dagger$). 
Second, the differences between the number of barcodes, $\tau_{max}$, and the number 
of queries in their method are no more than ten fold, but their weights are least 2.5 fold greater than the
weights in the light Chinese design. With the estimated cost ratio between barcode to 
sequencing lane to be around $1:100$, the light Chinese design is more cost effective.
%In fact, we noticed that the weight of other $t \sim O(d\log{n})$
%methods scale with $\log{n}$ (\hl{cite Haung}), and we conjecture that it is always  $w \sim  O(\log{n})$ 
%in those methods.
Finally, there is no intrinsic mechanism in Eppstein's method to specify the weight, 
and to limit it below a threshold.

\begin{table}[!t]
\centering
\caption{Comparison between Eppstein's Method to the Light Chinese Design}
\label{table_eppstein}
\begin{tabular}{|l|c|c|c|c|c|c|c|c|c|}
\hline
\multirow{2}{*}{$n$} &
\multirow{2}{*}{$d$} &
\multicolumn{4}{c|}{Eppstein} & 
\multicolumn{4}{c|}{Light Chinese design} \\ \cline{3-10}
& & $t$ & $w$ & $\tau_{max}$ & $r_{max}$ & $t$ & $w$ & $\tau_{max}$ & $r_{max}$ \\\hline
\multirow{3}{*}{5000} 	& 3 & 149 & $10^{\dagger}$ & 29 	& 1000	& 293 & 4 & 77 & 64\\
				    	& 4 & 237 & $12^{\dagger}$  & 37 	& 714	& 370 & 5 &  77 & 64\\
				    	& 5 & 336 & $15^{\dagger}$ & 47 	& 1000 	& 449 & 6 &  79 & 64 \\  \hline
\multirow{3}{*}{40000} & 3 & 209 & $12^{\dagger}$ & 37 	& $8000^{\dagger}$ & 811 & 4 & 205 & 199\\
					& 4 & 335 & $14^{\dagger}$ & 47 	& $5714^{\dagger}$ & 1020 & 5 & 209 & 199\\
					& 5 & 472 & $17^{\dagger}$ & 59 	& $5714^{\dagger}$ & 1231 & 6 & 211 & 199\\
\hline
\end{tabular}
\end{table}

\section{RECONSTRUCTION ALGORITHMS}\label{section_decoding}

\subsection{Bayesian Decoding Using Belief Propagation}
Now, we will turn to address the other part of the minimal genotyping problem, which is 
how to reconstruct $\mathbf{G}$ %(or just $E_{risk}$ 
given $\mathbf{Y}$ and $\mathbf{\Phi}$. 
In general, this is an ill-posed inverse problem, but the sparsity of $\mathbf{G}$ due to the biological 
constraints (e.g the diploidy of the genome and the absence of affected individuals in the screen)
and the low abundance of rare risk alleles permits such decoding.
The MAP decoding of the genotyping problem is given by:
\begin{equation}\label{equ_decoding_prob1}
	{\displaystyle \mathbf{G}_{MAP}
	\triangleq \argmax\displaylimits_{x_1,\dots, x_n} \Pr( x_1,\dots,x_n \mid \mathbf{Y})}
	%\prod\displaylimits_{i=1}^{n}{\mathbb{B}(x_i)}
\end{equation}
% $mathbb{B}(x_i)$ is an indicator function that asserts correct genome ploidy, and for human is:
% \begin{equation}\label{equ_what_is_b}
% 	{\displaystyle
% 		\mathbb{B}(x_i) =
% 			\left\{ 
% 				\begin{array}{rl}
% 				1 & \sum x_i = 2 \\
% 				0 & otherwise
%       				\end{array}
% 			\right.	
% 	}
% \end{equation}

For simplicity, we assume that we do not have any prior knowledge on the specimens, beside the 
expected frequency of the genotypes in the screen. Notice that kinship information between 
the specimens and familial history regarding genetic diseases are known is some cases and 
may enhance the decoding results, however, they will remain outside the scope of this manuscript.
$\eta$ is a probability vector with length of $s(s+1)/2$ that denotes 
the expected prevalence of each genotype. 
For instance, for $\Delta F508$ screen $\eta=(29/30,1/30,0)$ for normal, carriers, and affected,
 correspondingly.
Let $x_i$ be an instance of row vector in $\mathbf{G}$, and $c(x_i)$ be a binary vector
with length as $\eta$ that maps the allelic configuration of $x_i$ to an entry in a list of  genotypes. 
For instance, if there are two alleles in the population, $x_i$ is either $(2,0)$, $(1,1)$, or $(0,2)$,
and $c(x_i)$ is $(1,0,0)$, $(0,1,0)$, or $(0,0,1)$, correspondingly.
We denote the prior probability of $x_i$  by:
\begin{equation}\label{equ_define_phi}
	{\displaystyle
		\varphi(x_i) = 
		\prod\displaylimits_{j=1}^{s(s+1)/2}
		\eta_{j}^{c_j(x_i)}
	}	
\end{equation}
The prior probability for a certain graph configuration,  $(x_1, \dots, x_n)$, is:
\begin{equation}\label{equ_decoding_prior}
	{\displaystyle
		\prod\displaylimits_{i=1}^{n}
		\varphi(x_{i})
	}
\end{equation}

The data is also a subject to factorization, since the result of a particular query is solely determined by
the specimens in the pool:
\begin{equation}\label{equ_decoding_data_factorization}
	{\displaystyle
		\Pr(\mathbf{Y} \mid x_1, \dots , x_n) =
		\prod\displaylimits_{a=1}^{t}{\Pr(\mathbf{Y}_{a} \mid x_{\partial a} ) }
	}
\end{equation}
we used $x_{\partial a}$ to denote a configuration of the subset of specimens in the $a$ query, and 
$\mathbf{Y}_{a}$ denotes the $a^{th}$ row vector in $\mathbf{Y}$. The probability distribution
$\Pr(\mathbf{Y}_{a} \mid x_{\partial a} )$ is given by the compositional channel model in 
Eq. (\ref{equ_compositional_channel_with_errors}) and since we assume that $\beta$ and $\Lambda$
are constant for all the queries, we will use the following shorthand to denote this probability distribution:
\begin{equation}\label{equ_decoding_potential}
	{\displaystyle
		\Psi_{a}(x_{\partial a}) \triangleq f_{\beta}( x_{\partial a}\Lambda)
	}
\end{equation}
From Eq. (\ref{equ_decoding_prob1}-\ref{equ_decoding_potential}), 
we get:
\begin{equation}\label{equ_decoding_postpriorl}
	{\displaystyle 
		\Pr(\mathbf{G}) \propto
		\prod\displaylimits_{a=1}^{t}\Psi_{a}(x_{\partial a})
		\prod\displaylimits_{i=1}^{n}\varphi(x_{i})
%		\prod\displaylimits_{i=1}^{n}{\mathbb{B}(x_i)}
	}
\end{equation}

The factorization above is captured by factor graph with 
two types of factor nodes, $\varphi$ nodes and $\Psi$ nodes.
The $\varphi$ nodes are uniquely connected to each variable nodes, whereas the $\Psi$ nodes
are connected to the variables according to the query design in $\Phi$, so each variable node
is connected to $w$ different $\Psi$ nodes. An example of a factor graph with 12 specimens, and 
$Q=\{3,4\}$ is given in Fig. \ref{fig_factor_graph}.

\begin{figure}[!t]
\centering
\includegraphics[width=2in]{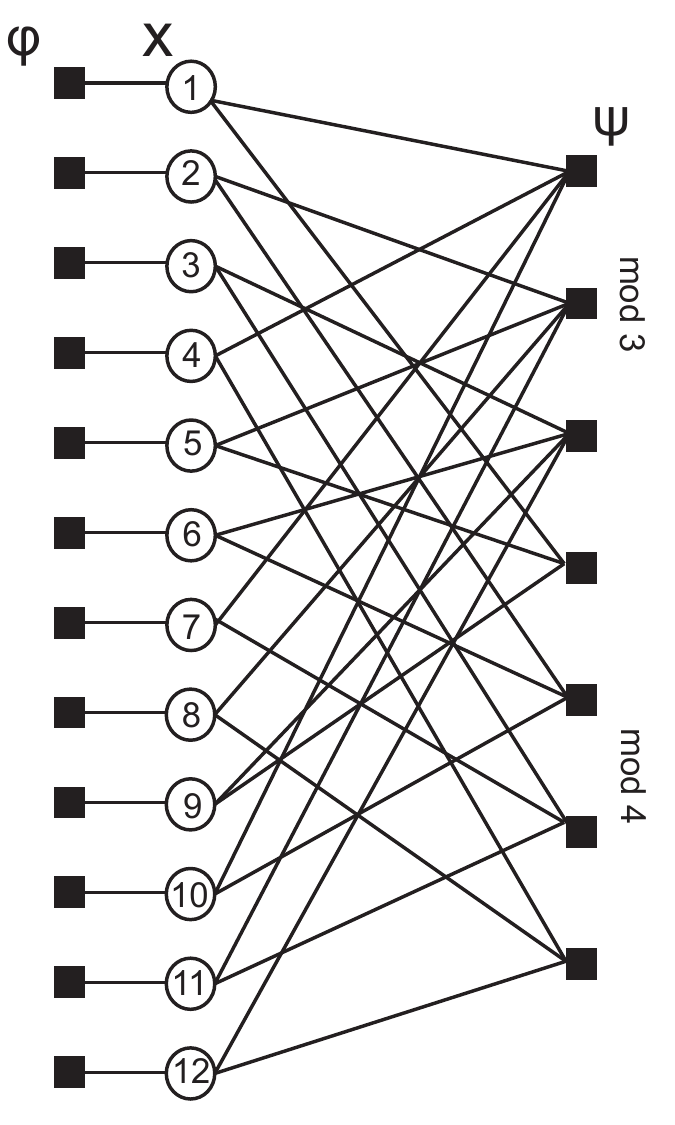}
\caption{Example of Factor Graph for Genotyping Reconstruction}
\label{fig_factor_graph}
\end{figure}

Belief propagation (sum-product algorithm)\cite{article_factor_graph, book_pearl} is a graphical inference technique that is based on exchanging messages (beliefs) between 
factor nodes and variable nodes that tune the marginals of the variable 
nodes to the observed data. When a factor graph is a tree
the obtained marginals are exact; however, a factor graph that is built according to any reasonable 
query design will always contain many loops (easily proved by the pigeonhole principle), implying that finding $G_{MAP}$ is NP-hard\cite{book_bishop}.
Surprisingly, it has been found that belief propagation can still be used as an approximation method 
for factor graphs with loops. These findings rely on the concept 
that if the local topology of a factor graph is a tree-like, 
the algorithm can still converge with high probability\cite{article_loopy_belief_propagation, book_Mezard}.
This approach has been successfully used in a broad spectrum 
of NP-hard problems including decoding LDPC codes\cite{article_factor_graph}, finding assignments in  random k-SAT problems\cite{article_ashish} and even
solving Sudoku puzzles \cite{article_sudoku}. Recently, Mezard and colleagues studied the decoding
performance of belief propagation in the prototypical problem of group testing\cite{atricle_mezard_paper}.
One advantage of their setting is the presence of 
'sure zeros' - variables nodes that are connected to 
at least one 'inactive' test node. Since the tests are faultless in the prototypical problem,
those variables are immediately decoded as 'inactive', and are stripped off from the factor graph,
which reduces the complexity of problem handed to the belief propagation.
Unfortunately, the query results from next generation sequencers are not reliable, and the absence of an allele node
from a query may stem from insufficient sequencing coverage (small $\beta$) and sequencing errors. Furthermore,
the total number of sure zeros can be very small as confounding errors may falsely indicate the 
presence of an allele in a query. From these two reasons, stripping has little applicability in our setting.
On the other hand, Baron and colleagues \cite{article_baron_belief} investigated the performance
of belief propagation for the recovery of compressed signals with a linear channel model and 
additive white Gaussian noise (AWGN). Our approach is reminiscent of their 
method, and employs belief propagation on the full graph using some essential shortcuts.

The marginal probability of $x_i$ is given by the Markov property of the factor graph:
\begin{equation}\label{equ_margin_of_x}
	{\displaystyle
		\Pr(x_i) \propto 
		\varphi(x_i)
		\prod\displaylimits_{a=1}^{w} 
		\mu_{a\rightarrow x_i}(x_i)
	}
\end{equation}
The approximation made by belief propagation in loopy graphs is that the beliefs of
the variables in the subset $\partial a \backslash x_i$ regarding $x_i$ are independent.
Since $\lambda_{max}=1$ in light-weight designs, the resulted factor graph 
does not have any short cycles of girth $4$, implying that the beliefs does not strongly correlated,
and that the assumption is approximately fullfilled.
The algorithm defines $\mu_{a\rightarrow x_i}(x_i)$ as:
\begin{equation}\label{equ_message_factor_to_node}
	{\displaystyle
	\mu_{a\rightarrow x_i}(x_i)=
	\sum\displaylimits_{\{x \in \partial a \backslash x_i\}}
	\Psi_{a}(x_{\partial a})
	\prod\displaylimits_{x_j \in \partial a \backslash x_i} \mu_{x_j\rightarrow a}(x_j)
	}
\end{equation}
and
\begin{equation}\label{equ_message_node_to_factor}
	{\displaystyle
		\mu_{x_j\rightarrow a}(x_j) =
		\varphi(x_j)
		\prod\displaylimits_{u \in \partial x_j \backslash a}
		\mu_{u \rightarrow x_j}(x_j)
	}
\end{equation}
were $u \in \partial x_j$ denotes the subset of queries with $x_j$.
Eq. (\ref{equ_message_factor_to_node}) describes message from a factor node to a variable node,
and Eq. (\ref{equ_message_node_to_factor}) describes message from a variable node to a factor node.
By iterating between the messages the marginals of the 
variable nodes are gradually obtained, and in case of successful decoding the algorithm reaches to a stable 
point, and reports $G^{*}$:
\begin{equation}\label{equ_g_star}
	{\displaystyle
		G^{*} \triangleq
		\argmax\displaylimits_{x_i} 
		\Pr(x_i)
	}
\end{equation}

This approach encouters a major obstacle - calculating the factor to node messages 
requires summing over all possible genotype configurations in the pool, which exponentially grows
with the compression level, $r_{max}$, or $\sqrt{n}$.
To circumvent that, we use Monte-Carlo sampling instead of an exact calculation to find the factor 
to node messages of each round. This is based on drawing random configurations of $x_{\partial a}$
according to the probability density functions (pdf) that are given by the
$\mu_{x_j\rightarrow a}(x_j)$ messages and evaluating $\Psi_a(x_{\partial a})$. 
An additional complication are strong oscillations in which the marginal estimation of $x_i$ for the $\tau$ step
is almost completely concentrated in one state, but at the $\tau+1$ step, the estimation is completely
concentrated in another state. One of those states is obviously wrong, and a sampling process that uses
this pdf to evaluate a factor to node message for other variable nodes may find only very 
small values of $\Psi$, which is prone to numerical stability issues that ended up in sending all-zero 
messages and failure of the algorithm.
We used message damping to attenuate the oscillations \cite{article_message_damping}. 
The damping procedure averages the variable to factor messages of the $m$ round with the message 
of the $m-1$ round:
\begin{equation}\label{equ_damping}
	{\displaystyle
	\mu_{x_j\rightarrow a}^{m(damped)}(x_j) =
	\left(\mu_{x_j\rightarrow a}^{m}(x_j)\right)^{1-\gamma}
	\left(\mu_{x_j\rightarrow a}^{m-1}(x_j)\right)^{\gamma}
	}
\end{equation}
The extent of the damping can by tuned with $\gamma \in [0,1]$. When $\gamma=1$ 
there are no updated at all, and when $\gamma=0$ we restore the algorithm in Eq. (\ref{equ_message_node_to_factor}).
Appendix \ref{app_full_lay} presents a full layout of the belief propagation reconstruction algorithm:

\subsection{Baseline Reconstruction Algorithm}
In order to benchmark the belief propagation decoding algorithm above, we introduce
an additional algorithm, named {\it pattern consistency decoding}, which is used in group testing
to reconstruct the original data from superimposed channel. In a carrier screen, the algorithm 
first creates a new matrix that is composed of the columns in $\mathbf{Y}$
that correspond to the risk alleles, and then it treats the results in the new matrix
as superimposition according to Eq.
(\ref{superimposed_degredation}). We denote the new matrix by $\mathbf{Y}_{rs}$.

This method does not address query errors, and a specimen is defined as a carrier only if all its $w$ queries indicate the presence of a risk allele:
\begin{equation}\label{equ_pattern_consistency_decoder}
	\widehat{\mathbf{E}_{risk}} = 
	\mathbb{I}(
		\mathbf{Y}^{T}_{s}
		\mathbf{\Phi}
	)
\end{equation}
where $\mathbb{I}$ is an indicator function:
\begin{equation}\label{equ_what_is_1}
 	{\displaystyle
 		\mathbb{I}(X_{ij}) =
 			\left\{ 
 				\begin{array}{rl}
 				1 & X_{ij} = w \\
 				0 & otherwise
      				\end{array}
 			\right.	
 	}.
 \end{equation}
Rows of $E_{risk}$ with positive entries indicate carriers.
This reconstruction is guaranteed to be correct if $d_0$, the maximal number of 
carriers in the screen for one of the risk alleles, is lower than $d$, the disjunction property of $\mathbf{\Phi}$, 
(given no sequencing errors and sufficient coverage).
Since this reconstruction works with degraded information compared to belief propagation, we will use it to indicate the baseline performance expected from belief propagation decoding, and to test whether the approximations we employed (loopy messages, Monte-Carlo sampling, damping) are valid.

\section{Numerical Results}\label{section_results}
\begin{figure*}[!t]
\centering
\includegraphics[width=6in]{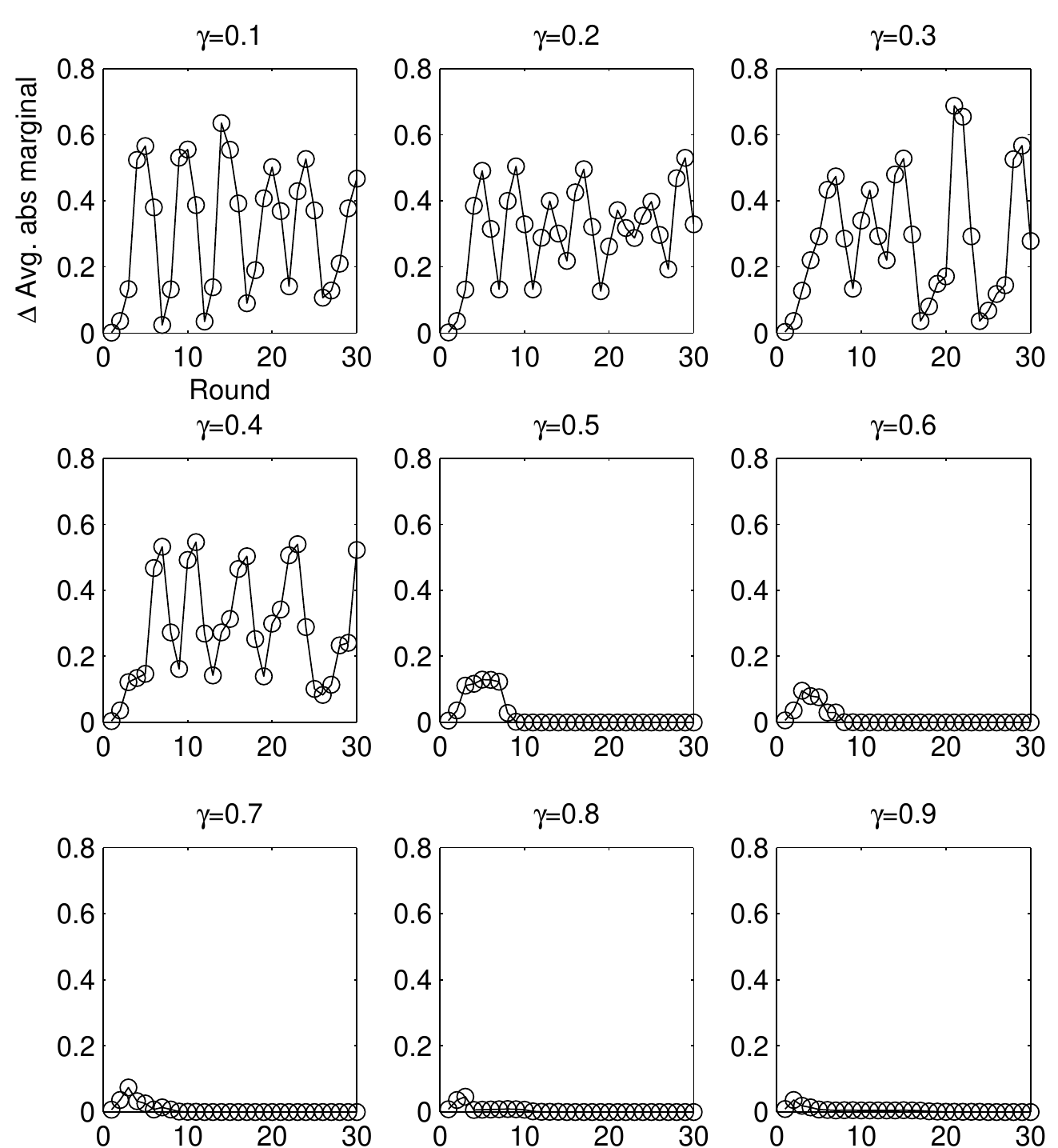}
\caption{The Effect of Damping on Oscillations}
\label{fig_gamma}
\end{figure*}

To demonstrate the power of our method, we simulated several settings where 
there is one risk allele and one WT allele in the population, with $n=1000$, $\beta = 10^{3}$, $w=5$, 
and $Q=\{33, 34, 35, 37, 41\}$, which can be accommodated in a single 
machine batch. Fig. \ref{fig_gamma} emphasizes the 
effect of damping on the belief propagation convergence rates. In this example, the 
number of carriers in the screen was $d_{0}=43$, and we ran the decoder for $30$ iterations.
We evaluated different extents of damping: $\gamma \in [0.1, \dots, 0.9]$, and we measured for each 
iteration the averaged absolute difference in the marginal from the previous step. 
We found that with $\gamma<0.5$, there are strong oscillations and the algorithm 
does not converge, whereas
 with $\gamma \geqslant 0.5$, there are no oscillations, and the algorithm converges and 
correctly decodes the genotype for all the specimens.

We also tested the performance of the reconstruction algorithms for increasing number of carriers 
in the screen, ranging from $5$ to $150$, 
with no sequencing errors  (Fig. \ref{fig_d}).
%In that simulation the number of samples for each iteration was $2500$, and we let the algorithm to run for
%$50$ cycles. 
The belief propagation reconstruction outperformed the pattern consistency decoder and reconstructed 
the genotypes with no error even when the number of carriers was $40$, which is a quite high number for 
severe genetic diseases. The ability of the belief propagation to faithfully reconstruct cases with
 $d_{0}\gg d$-disjunction of the query design is not suprising,  since d-disjunction is a conservative 
sufficient condition even for a superimposed channel.

\begin{figure}[!t]
\centering
\includegraphics[width=2.5in]{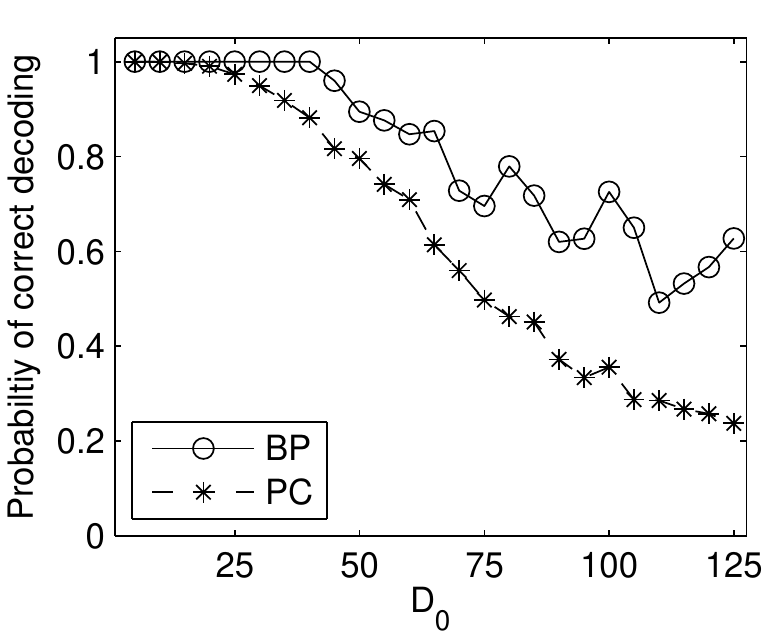}
\caption{Decodability as a Function of Number of Carriers}
\label{fig_d}
\end{figure}

We continue to evaluate the performance of the algorithm in a biologically-oriented setting - 
detecting carriers for CF $W1282X$ mutation, where the carrier rate in some populations is about $1.8\%$ \cite{article_cf_israeli}.
The relatively high rate of the carriers challenges our scheme with a difficult genetic screening
problem. Moreover, the sequence difference between the WT allele and the mutant allele is only a single
 base substitution, and sequencing error may cause genotype confounding.
To recapitulate that, we introduced increasing levels of symmetric 
confounding errors (i.e the two alleles
have the same probability of being converted from one to the other), and we tested the performance
of the reconstruction algorithms with $\beta=10^{3}$ and $\beta=10^{4}$, and with error rates in the
range of $0\%-4.5\%$ with steps of $0.5\%$ (Fig. \ref{fig_beta}).
\begin{figure}[!t]
\centering
\includegraphics[width=2.5in]{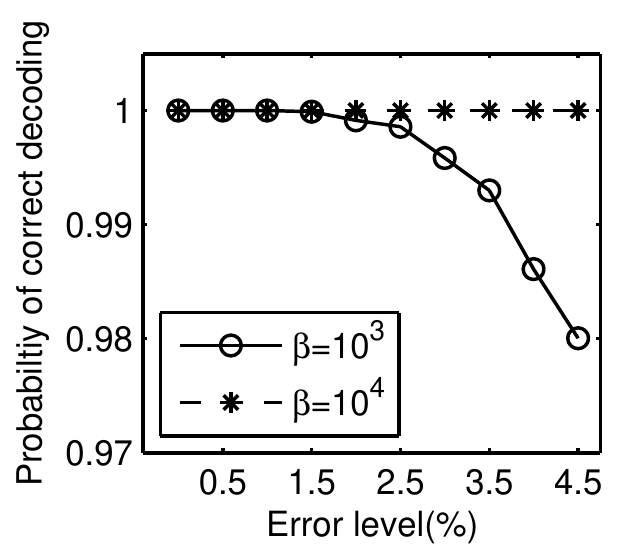}
\caption{Simulation of CF Screen - The Effect of $\beta$ and Confounding Errors}
\label{fig_beta}
\end{figure}

\begin{figure}[!t]
\centering
\includegraphics[width=2.5in]{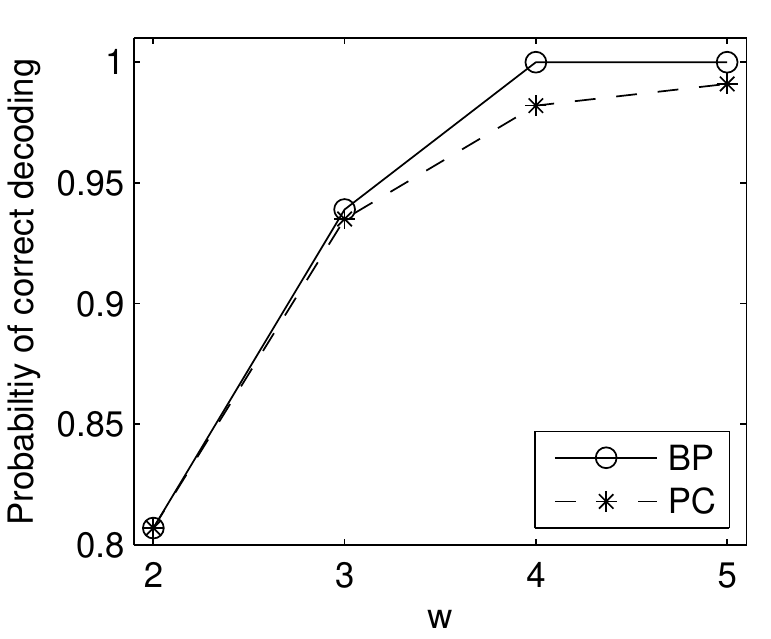}
\caption{Different Weights for CF $\Delta F508$ Screen Detection}
\label{fig_w}
\end{figure}

As expected, the pattern consistency decoder performed poorly (data not shown) even for the lowest 
error rate of $0.5\%$ and marked all specimens as carriers.
The belief propagation algorithm reported the correct genotype for all specimens even 
when the error rate was $1.5\%$ and $\beta=10^{3}$. Importantly, the decoding mistakes of the belief
propagation at higher error rates were false positives, and did not affect the sensitivity of the method.
When we increased the number of reads for each query to $\beta=10^{4}$, the belief propagation
decoder reported the genotype of all the specimens without any mistake. 
As we mentioned earlier, the expected confounding error rate for this mutation up to $1\%$, 
implying that the parameters used in the simulation are quite conservative.

We also tested another CF mutation, $\Delta F508$, which has a similar carrier rate in people with 
European descents as $W1282X$, but contains a $3$-nucleotide deletion when compared to the WT allele.
This implies that the confounding error rates are negligible, as sequencing-induced deletions are quite rare.
In this example, we evaluated the effect of different weights for the query designs, and
we used the following sets of query windows: $\{33, 34\}$, $\{33, 34, 35\}$, $\{33, 34, 35, 37\}$, 
$\{33, 34, 35, 37, 41\}$. Fig. \ref{fig_w} depicts the results for the belief propagation algorithm 
and for the pattern consistency decoder. While the results are quite poor for $w=2$, the belief propagation
 decodes correctly all the specimens with $w=4$, which would requires the synthesis of only $37$ barcodes, 
and a total of $139$ queries.

\section{CONCLUSION}\label{section_conclude}
In this paper, we presented a compressed genotyping framework that harnesses next generation 
sequencers for large scale genotyping screens of severe genetic diseases.
We formulated the problem as reconstructing
a spares bipartite multigraph from information that was obtained over a compositional channel. 
In addition to the traditional objective of minimizing the number of queries, we introduced another objective of
reducing the weight of the design, and we propose a new class of designs called light-weight designs
in which the weight does not depend on
 $n$, and only grows linearly with $d$. 
For the genotyping reconstruction part, we presented a Baysian framework that is based on 
loopy belief propagation, and we evaluated its performance by simulating different types of carrier tests,
including prevalent mutations in Cystic Fibrosis.

Further investigation is needed to expand the framework to include prior biological data such as 
familial information and other predispositions, and to include more types of errors beyond 
those introduced by sequencing, such as biased PCR amplification, query failures, and sample
contamination. In addition it will be interesting to develop a more comprehensive treatment for the
compositional channel, and to find a less conservative sufficient condition for faithful signal reconstruction.

\appendices
\section{The product maximization algorithm}\label{app_product_max}
The product maximization problem is defined as follows. Given parameters
$\kappa$, $\kappa_1$ and $w$, find the set $Q$ of size $w$ whose elements are all
in the range $\kappa<x \le \kappa_1$ and such that for no pair $x,y \in Q$ has
$lcm(x,y) \le \kappa^2$. For product maximization, typical values in practice
have $\kappa$ in the range $[100,300)$, $w$ in the range $[2,8]$ and
$\kappa_1$ fixed at $384$. The reason for this number is the number of wells in a microtiter plate,
which is compatible with liquid handling robots. 
The empirical results below relate to this entire range,
for all of which we have optimal solutions discovered by exhaustive searching.

The product maximization problem has ties to the sum minimization problem in
both bound-calculation and solving algorithms. First, note that in this
problem we cannot consider ``asymptotic'' behavior when $w$, $\kappa$ and $\kappa_1$
are large without specifying how the ratio $\frac{\kappa_1}{\kappa}$ is constrained.

If $\kappa$ is constant and $\kappa_1$ rises, the asymptotic solution will be the set
$\{\kappa_1,\kappa_1-1,\kappa_1-2,\ldots,\kappa_1+1-w\}$. 
This set clearly has the maximum possible
product, while at the same time satisfying the condition on the $lcm$ because
no two elements in the solution can have a mutual factor greater than $w$.
This value will be the optimum as soon as $(\kappa_1+1-w)(\kappa_1+2-w)>w \kappa^2$ (and
possibly even before), so $\sqrt{w}$ should be taken as an upper bound for
$\frac{\kappa_1}{\kappa}$ to form a non-trivial case.

For any specific ratio $\frac{\kappa_1}{\kappa}$, the condition $lcm(x,y) > \kappa^2$
for $x$ and $y$ values close to $\kappa_1$ is equivalent to
$gcd(x,y) < \frac{\kappa_1^2}{\kappa^2}$. This allows us to reformulate the question
as that of finding the set $Q$ with $w$ elements, all less than or equal to
$\kappa$, s.t. the $gcd$ of any pair is lower than or equal to
$\rho = \left\lfloor\frac{\kappa_1^2}{\kappa^2}\right\rfloor$.

For the product maximization problem, we redefine the discrepancy to be
$\delta=\kappa_1+1-\min(Q)$. In order to compute the asymptotic bound for this
discrepancy, let us first define
\emph{pseudo-primes}. Let the set of $k$-pseudo-primes, $P_k$, be defined as
the set s.t. $i \in P_k \iff i>k$ and $\neg \exists j<i,j \in P_k$ s.t.
$i$ is divisible by $j$. The set of $1$-pseudo-primes coincides with the set
of primes.

One interesting property of $k$-pseudo-primes is that they coincide with the
set of primes for any element larger than $k^2$. To prove this, first note
that if $i$ is a prime and $i>k$ then $i$ by definition belongs to $P_k$.
Second, note that if $i$ is composite and $i>k^2$ then $i$ has at least one
divisor larger than $k$. In particular, it must have a smallest divisor
larger than $k$, and this divisor cannot have any divisors larger than $k$,
meaning that it must belong to $P_k$. Consequently, $i \notin P_k$.

Both the reasoning yielding the upper bound and the reasoning yielding the
lower bound for the sum minimization problem utilize estimates for the density
of numbers not divisible by a prime smaller than some $d$. In order to fit
this to the product maximization problem, where a $gcd$ of $\rho$ is allowed,
we must revise these to estimates for the density of numbers not divisible by
a $\rho$-pseudo-prime smaller than $\delta$. Because the $k$-pseudo-primes and the
primes coincide beginning with $k^2$, this density is the same up to an
easy-to-calculate multiplicative constant $\gamma_k$.

Knowing this, both upper and lower bound calculations can be applied to
show that the asymptotic discrepancy of the optimal solution is on the order of
$\gamma_k w \ln(w)$. This discrepancy can be used, as before,
to predict an approximate optimal product. However, the bound on the product
is much less informative than the bound on the sum: the product can be bounded
from above by $\kappa_1^w$ and from below by $(\kappa_1-\delta)^w$, 
both converging to a ratio of 1:1 at $\kappa_1$ rises.

The revised greedy algorithm for this problem is given explicitly below.

\begin{algorithmic}[1]
\STATE Let $Q$ be the set of the $w$ largest primes $\le \kappa_1$.
\REPEAT
\STATE $\delta \gets \kappa_1+1-\min(Q)$
\STATE $Q \gets$ the $w$ largest numbers $\le \kappa_1$ that have no
factors smaller than $\delta$
\UNTIL {$\delta=\kappa_1+1-\min(Q)$}
\REPEAT
\STATE $\delta \gets \kappa_1+1-\min(Q)$
\STATE $n(x) \defeq$ the number of distinct primes smaller than $\delta$ in the
factorization of $x$.
\STATE Sort the numbers $\min(Q), \ldots, \kappa_1$ by increasing $n(x)$ [major key]
and decreasing value [minor key].
\FORALL {$i$ in the sorted list}
\IF {$\forall q \in Q, lcm(q,i)>\kappa^2$ and $i>\min(Q)$}
\STATE replace $\min(Q)$ by $i$ in $Q$.
\ELSIF {There is exactly one $q \in Q$ s.t. $lcm(q,i)\le \kappa^2$, and $i>q$}
\STATE replace $q$ with $i$ in $Q$.
\ENDIF
\ENDFOR
\UNTIL {$\delta=\kappa_1+1-\min(Q)$}
\STATE output $Q$.
\end{algorithmic}

Note that the greedy algorithm tries to lower the discrepancy of the solution
even when there is no proof that a smaller discrepancy will yield an improved
solution set. In the sum minimization problem, any change of $\Delta$ in any
of the variables yields a change of $\Delta$ in the solution, so there is
little reason to favor reducing the largest element of $Q$ (and thereby
reducing the discrepancy) over reducing any other element of $Q$. In product
maximization, however, a change of $\Delta$ to $\min(Q)$ (and hence to the
discrepancy) corresponds to a larger change to the product than a change of
$\Delta$ to any other member of $Q$. This makes the greedy algorithm even
more suited for the product maximization problem than for sum minimization.

Indeed, when examining the results of the greedy algorithm on $\kappa=384$, with
$w \in [2,8]$ and $\kappa \in [100,300)$ we see that the greedy algorithm
produces the correct result in all cases $w\in [2,3,4]$. In $w\in [6,7,8]$
the algorithm produces the optimal result in all but 2,3 and 3 cases,
respectively. The only $w$ for which a large number of sub-optimal results
was recorded is $w=5$ where the number of sub-optimal results was $49$. Note,
however, that in product maximization there is a much larger tendency for
``streaking''. The $49$ sub-optimal results all belong to a single streak,
where the optimal answer is $\{379, 380, 381, 382, 383\}$ and the answer
returned from the greedy algorithm is $\{377, 379, 382, 383, 384\}$. The
difference in the two products is approximately $0.008\%$.

In terms of streaks, the optimal answer was returned in all but one streak
in $w\in[5,6]$ and in all but two streaks in $w\in[7,8]$. In terms of the
number of iterations required, the only extra iterations that were
needed in the execution of the algorithm beyond the minimal required was a
single extra iteration through the
first ``repeat'' loop when $w$ was $3$. In all other cases, no extra iterations
were used, demonstrating that this algorithm is in practice faster than is
predicted by its (already low-degree polynomial) time complexity.

\section{Full Layout of Belief Propagation Reconstruction}\label{app_full_lay}
\begin{enumerate}
\item{\bf Inputs:} Query design $\mathbf{\Phi}$, sequencing results $\mathbf{Y}$,  prior expectations about the
genotypes prevalence $\eta$, damping parameter $\gamma$, 
number of iterations $m_{max}$, and number of Monte Carlo rounds $z$.

\item {\bf Preprocessing:} 
(a) find $\beta$ - enumerate the number of reads in the query.
(b) learn the genotype error pattern $\Lambda$ - the sequencing errors
rates are estimated using spiked-in controls \cite{article_me_nm}, and 
converted to genotype error according to the sequence of the different alleles.
(c) calculate $\varphi$ according to $\eta$.

\item {\bf Initialization} Initialize the iteration counter $m$. Initialize $\mu_{x_i\rightarrow a}(x_i)$ to 
priors in $\varphi$.

\item {\bf Send messages from factors to variables:}\label{ref_algorithm_1}
\begin{algorithmic}[1]
\FOR {each factor $a$ in $\{1, \dots, t\}$}
	\FOR {each variable $x_i$ in query $a$}
		\FOR {each state of variable $x_i$ in $\{1, \dots , |\eta|\}$}
			\STATE Set $\Psi_{0} \gets 0$
				\FOR {$\{1, \dots, z\}$ Monte-Carlo round}
					\STATE  $r \gets$ random configuration of $\partial a \backslash x$ 
					according to pdfs in $\mu_{x_j\rightarrow a}^{m}$
					\STATE  $\Psi_{0} \gets$ $\Psi_{0}+\Psi_{a}(r,$ state of $x_i)$
				\ENDFOR
				\STATE $\mu_{a\rightarrow x_i}($state of $x_i) \gets \Psi_{0}/m$
		\ENDFOR
		\STATE Normalize $\mu_{a\rightarrow x_i}(x_i)$
		\STATE Send message $\mu_{a\rightarrow x_i}(x_i)$
	\ENDFOR
\ENDFOR
\end{algorithmic}
\item {\bf Send messages from variables to factors:}
\begin{algorithmic}[1]
\FOR {each variable $x_i$ in $\{1, \dots, n\}$}
	\FOR {each factor $a$ connected to $x_i$}
		\STATE Set $\mu_{x_i\rightarrow a}^{m}(x_i)$ to all ones vector.
		\FOR {each possible state of variable $x_i$ in $\{1, \dots , |\eta|\}$}
			\FOR {each factor $j$ connected to $x_i$ except $a$}
				\STATE $\mu_{x_i\rightarrow a}^{m}($ state of $x_i)$ =
					       $\mu_{x_i\rightarrow a}^{m}($ state of $x_i)$ 
				               $\mu_{j\rightarrow x_i}($ state of $x_i)$
			\ENDFOR	
		\ENDFOR
		\STATE Include prior by 
		$\mu_{x_i\rightarrow a}^{m}(x_i) \gets \mu_{x_i\rightarrow a}^{m}(x_i)\varphi(x_i)$
		\STATE Damp $\mu_{x_i\rightarrow a}^{m}(x_i)$
		\STATE Normalize $\mu_{x_i\rightarrow a}^{m}(x_i)$
		\STATE Send message $\mu_{x_i\rightarrow a}^{m}(x_i)$
	\ENDFOR
\ENDFOR
\STATE $m \gets m +1$
\end{algorithmic}
Go back to step \ref{ref_algorithm_1} if $m < m_{max}$.

\item {\bf Marginalize:} For every variable node compute the marginal according to Eq.
(\ref{equ_margin_of_x}), and find the state of the variable with the highest probability.

\item {\bf Report:} Report the highest state of each variable and construct $\mathbf{G}$.
\end{enumerate}

\section*{Acknowledgment}
The authors thank Oded Margalit and Oliver Tam for useful comments.
Y.E is a Goldberg-Lindsay Fellow and ACM/IEEE Computer Society High Performance  Computing
PhD Fellow of the Watson School of Biological Sciences. G.J.H is an investigator of the Howard Hughes
Medical Institute. P.P.M is a Crick-Clay Professor.

\bibliographystyle{hieeetr}
\bibliography{compressed_sensing}
\end{document}